%% file: main.tex
\documentclass[runningheads,orivec,envcountsame]{llncs}

\usepackage[usenames,dvipsnames]{xcolor}
\usepackage{centernot}
\usepackage{amssymb}
\usepackage[linesnumbered,noend,boxruled]{algorithm2e}
\usepackage{stmaryrd}
\usepackage{bm}
\usepackage{tikz}
\usetikzlibrary{shapes,calc,arrows,automata,arrows.meta}
\usepackage{multirow}
\usepackage{array}
\usepackage{mathtools}
\usepackage{enumitem}

\makeatletter
\renewcommand{\nllabel}[1]
 {{\let\@currentlabel\algocf@currentlabel
  \let\@currentcounter\algocf@currentcounter
  \label{#1}}}%

\renewcommand{\algocf@nl@sethref}[1]{%
  \renewcommand{\theHAlgoLine}{\thealgocfproc.#1}%
  \hyper@refstepcounter{AlgoLine}%
  \gdef\algocf@currentlabel{#1}%
  \gdef\algocf@currentcounter{AlgoLine}%
 }%
\makeatother

\usepackage[bookmarks,unicode,colorlinks=true]{hyperref}
\usepackage{proof}
\usepackage{varioref}
\usepackage[capitalize,nameinlink]{cleveref}
\usepackage[location=appendix,manual]{moveproofs}
\usepackage{thm-restate}
\usepackage{subfig}
\usepackage{hhline}
\usepackage{graphicx}
\usepackage{microtype}
\usepackage{cite}
\usepackage{pgfplots}
\usepackage{listings}
\pgfplotsset{compat=1.18}
\usepackage[edges]{forest}
\usepackage{empheq}
\usepackage[font=small,skip=0pt]{caption}
\usepackage{wrapfig}
\usepackage{marvosym}

\makeatletter
\newcommand{\algorithmstyle}[1]{\renewcommand{\algocf@style}{#1}}
\newcommand{\nosemic}{\renewcommand{\@endalgocfline}{\relax}}
\newcommand{\dosemic}{\renewcommand{\@endalgocfline}{\algocf@endline}}
\let\oldnl\nl
\newcommand{\nonl}{\renewcommand{\nl}{\let\nl\oldnl}}
\makeatother

\newcommand{\report}[1]{#1}
\newcommand{\paper}[1]{}

\paper{
  \setlength{\intextsep}{3pt}
  \AtBeginDocument{
    \addtolength\abovedisplayskip{-0.3\baselineskip}
    \addtolength\belowdisplayskip{-0.3\baselineskip}
    \addtolength\abovedisplayshortskip{-0.3\baselineskip}
    \addtolength\belowdisplayshortskip{-0.3\baselineskip}
  }
}

\SetKwIF{If}{ElseIf}{Else}{if}{do}{else if}{else}{end if}%
\SetKwFor{While}{while}{}{end}%

\SetCommentSty{mycommfont}

\DontPrintSemicolon

\hypersetup{%
  pdftitle={Infinite State Model Checking by Learning Transitive Relations},
  colorlinks=true,
  linkcolor=blue,
  citecolor=olive,
  filecolor=magenta,
  urlcolor=cyan
}

\makeatletter
\RequirePackage[bookmarks,unicode,colorlinks=true]{hyperref}%
   \def\@citecolor{blue}%
   \def\@urlcolor{blue}%
   \def\@linkcolor{blue}%

\def\orcidID#1{\href{http://orcid.org/#1}{\protect\raisebox{-1.25pt}{\protect\includegraphics{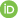}}}}
\makeatother

\newenvironment{proofsketch}{%
  \proof}{\endproof}

\renewcommand{\epsilon}{\varepsilon}
\let\oldphi\phi
\let\oldvarphi\varphi
\renewcommand{\varphi}{\oldphi}
\renewcommand{\phi}{\oldvarphi}

\newcommand{\pl}[1]{\textsf{#1}}
\newcommand{\tool}[1]{\textsf{#1}}
\def\increment{\hspace{-.05em}\raisebox{.4ex}{\tiny\bf ++}}
\def\decrement{\hspace{+.05em}\raisebox{.4ex}{\tiny\bf {-}{-}}}
\def\CXX{{\pl{C}\nolinebreak[4]\increment}}

\renewcommand{\vec}[1]{\bm{\mathrm{#1}}}

\newcommand{\rec}{\mathsf{rec}}
\newcommand{\push}{\mathsf{push}}
\newcommand{\pop}{\mathsf{pop}}
\newcommand{\add}{\mathsf{add}}

\newcommand{\ti}{\mathsf{ti}}
\newcommand{\accel}{\mathsf{accel}}
\newcommand{\underapprox}{\mathsf{ua}}
\newcommand{\tail}{\mathsf{tail}}

\newcommand{\Loop}{\mathsf{loop}}
\newcommand{\inc}{\mathsf{inc}}
\newcommand{\dec}{\mathsf{dec}}
\newcommand{\blocked}{\textsc{blocked}}
\newcommand{\blockingclause}{\mathsf{blocking\_clause}}

\newcommand{\checksat}{\mathsf{check}\_\mathsf{sat}}
\newcommand{\getmodel}{\mathsf{model}}
\newcommand{\mbp}{\mathsf{cvp}}
\newcommand{\mbip}{\mathsf{cvp}}
\newcommand{\sip}{\mathsf{sip}}

\newcommand{\unsafe}{\mathsf{unsafe}}
\newcommand{\unknown}{\mathsf{unknown}}
\newcommand{\encode}{{\sf encode}}

\newcommand{\safe}{\mathsf{safe}}
\newcommand{\tip}{\mathsf{tp}}
\renewcommand{\AA}{\mathcal{A}}

\newcommand{\CC}{\mathcal{C}}

\newcommand{\VV}{\mathcal{V}}

\newcommand{\trace}{\mathsf{trace}}
\newcommand{\compose}{{\mathsf{compose}}}
\newcommand{\concat}{\mathrel{::}}
\newcommand{\init}{\mathsf{init}}
\newcommand{\err}{\mathsf{err}}
\newcommand{\QF}{\mathsf{QF}}
\renewcommand{\partial}{\rightharpoonup}

\def\mystack#1\over#2_#3{%
   \mathrel {%
      \setbox0=\hbox{$\scriptscriptstyle #1$}%
      \setbox1=\hbox{$#2$}%
      \ifdim\wd1>\wd0 \kern .5\wd1 \else \kern .5\wd0 \fi
      \vbox{
         \offinterlineskip
         \moveleft.5\wd0 \box0
         \kern.3ex
         \moveleft.5\wd1 \hbox{$#2_#3$}
}}}

\newcommand{\ind}[3][]{
  \ifthenelse{\equal{#1}{}}{\overset{\scriptscriptstyle(#3)}{#2}}{{\mystack (#3) \over #2_#1}}
}

\newcommand{\ZZ}{\mathbb{Z}}

\newcommand{\NN}{\mathbb{N}}

\newcommand{\TT}{\mathcal{T}}

\newcommand{\Def}{\mathrel{\mathop:}=}

\renewcommand{\emptyset}{\varnothing}

\def\me{JG}
\usepackage{ifthen}
\newcommand{\comment}[2][ALL]{%
  \ifthenelse{\equal{ALL}{#1}}%
  {\footnote{!!! #2}}%
  {%
    \ifthenelse{\equal{\me}{#1}}%
    {\footnote{!!! #2}}%
    {}%
  }%
}

\DeclareMathOperator{\dom}{dom}

\newcommand{\id}{\mathsf{id}}

\crefname{algorithm}{alg.}{algorithms}%
\crefname{equation}{eq.}{equations}%
\crefname{chapter}{chapter}{chapters}%
\crefname{section}{sect.}{sections}%
\crefname{appendix}{app.}{appendices}%
\crefname{enumi}{item}{items}%
\crefname{footnote}{footnote}{footnotes}%
\crefname{figure}{fig.}{figures}%
\crefname{table}{table}{tables}%
\crefname{theorem}{thm.}{theorems}%
\crefname{lemma}{lemma}{lemmas}%
\crefname{corollary}{cor.}{corollaries}%
\crefname{proposition}{proposition}{propositions}%
\crefname{definition}{def.}{definitions}%
\crefname{result}{result}{results}%
\crefname{example}{ex.}{examples}%
\crefname{remark}{remark}{remarks}%
\crefname{note}{note}{notes}%
\crefname{lstlisting}{listing}{listings}%
\crefname{requirement}{req.}{requirements}%

\title{Infinite State Model Checking by Learning Transitive Relations}
\author{Florian Frohn$^{(\href{mailto:florian.frohn@informatik.rwth-aachen.de}{\mbox{\Letter}})}$\orcidID{0000-0003-0902-1994} and Jürgen Giesl$^{(\href{mailto:giesl@informatik.rwth-aachen.de}{\mbox{\Letter}})}$\orcidID{0000-0003-0283-8520}}
\paper{\institute{RWTH Aachen University,  Aachen, Germany}}
\report{\institute{RWTH Aachen University,  Aachen, Germany\\
\email{\{florian.frohn,giesl\}@informatik.rwth-aachen.de}}}
\authorrunning{F.\ Frohn, J.\ Giesl}

\begin{document}

\renewcommand{\thelstlisting}{\arabic{lstlisting}}

\maketitle

\input{abstract}

\input{introduction}
\paper{
  \input{overview_short}
}
\report{
  \input{overview}
}
\input{preliminaries}

\input{til}
\input{recurrence}

\report{
  \input{unsafety}
  \input{related}

  \input{experiments}

}
\paper{
  \input{related_short}
  \input{experiments_short}
}

\bibliographystyle{splncs04}
\paper{
  \bibliography{refs,crossrefs,strings}
}
\report{
  \bibliography{refs,crossrefs,strings}
  \clearpage 
  \input{proofs}

}

\end{document}

%% file: abstract.tex
\begin{abstract}
  We propose a new approach for proving safety of infinite state systems.
  It extends the analyzed system by \emph{transitive
  relations} until its \emph{diameter} $D$ becomes finite, i.e., until constantly many
  steps suffice to cover all reachable states, irrespective of the
  initial state.
  Then we can prove safety by checking that no error state is reachable in $D$ steps.
  To deduce transitive relations, we use \emph{recurrence analysis}.
  While recurrence analyses can usually find conjunctive relations only, our approach also discovers disjunctive relations by combining recurrence analysis with \emph{projections}.
  An empirical evaluation of the implementation of our approach in our tool \tool{LoAT}
  shows that it is highly competitive with the state of the art.
\end{abstract}

%% file: introduction.tex
\section{Introduction}
\label{sec:intro}
\label{sec:overview}

We consider relations defined by SMT formulas over two disjoint vectors of \emph{pre-} and \emph{post-variables} $\vec{x}$ and $\vec{x}'$.
Such \emph{relational formulas} can easily represent \emph{transition systems}
(TSs), linear \emph{Constrained Horn Clauses} (CHCs), and \emph{control-flow automata}
(CFAs).\footnote{To this end, it suffices to introduce one additional variable that represents the control-flow location (for TSs and CFAs) or the predicate (for linear CHCs).}
Thus, they subsume many popular intermediate representations used for verification of systems specified in more expressive languages.

In contrast to, e.g., source code, relational formulas are unstructured.
However, source code may be unstructured, too (e.g., due to {\tt goto}s), so being independent from the structure of the input makes our approach broadly applicable.

\begin{example}[Running Example]
  \label{ex:ex1}
  Let $\tau \Def \tau_{\inc} \lor \tau_{\dec}$ with:
  \begin{align*}
    w \doteq 0 \land x' \doteq x + 1 \land y' \doteq y + 1 \tag{$\tau_\inc$} \\
    w' \doteq w \land w \doteq 1 \land x' \doteq x - 1 \land y' \doteq y - 1 \tag{$\tau_\dec$}
  \end{align*}
  We use ``$\doteq$'' for equality in relational formulas.
  The formula $\tau$ defines a relation $\to_{\tau}$ on $\ZZ \times \ZZ \times \ZZ$ by relating the non-primed pre-variables with the primed post-variables.
  So for all $v_w,v_x,v_y, v'_w,v'_x, v'_y \in \ZZ$, we have $(v_w,v_x,v_y) \to_{\tau} (v'_w,v'_x, v'_y)$ iff $[w/v_w,x/v_x,y/v_y,w'/v'_w,x'/v'_x,y'/v'_y]$ is a model of $\tau$.
  Let the set of \emph{error states} be given by $\psi_{\err} \Def w \doteq 1 \land x \leq 0 \land y > 0$.
\end{example}
With the \emph{initial states} $\psi_{\init} \Def x \doteq 0 \land y \doteq 0$ this example\footnote{\href{https://github.com/chc-comp/chc-comp23-benchmarks/blob/main/LIA-lin/chc-LIA-Lin_005.smt2.gz}{\tt extra-small-lia/bouncy\_symmetry} from the \href{https://github.com/orgs/chc-comp/repositories}{CHC competition}}
is challenging for existing model checkers:
Neither the default configuration of \tool{Z3/Spacer} \cite{z3,spacer}, nor \tool{Golem}'s
\cite{golem} implementation of Spacer, LAWI \cite{lawi}, IMC \cite{imc}, TPA \cite{tpa-multiloop},
PDKIND \cite{pdkind}, or predicate abstraction \cite{predicate_abstraction} can prove its safety.
In contrast, all of these techniques can prove safety with the more general initial states $\psi_{\init} \Def x \doteq y$.
As all of them are based on \emph{interpolation}, the reason might be that the inductive invariant $x \doteq y$ is now a subterm of $\psi_{\init}$, so it is likely to occur in interpolants.
However, this explanation is insufficient, as all techniques fail again for $\psi_{\init} \Def x \doteq y \land y \doteq 0$.

This illustrates a well-known issue of interpolation-based verification techniques:
They are highly sensitive to minor changes of the input or the underlying interpolating SMT solver (e.g., \cite[p.~102]{gspacer}).
So while they can often solve difficult problems quickly, they sometimes fail for easy examples like the one above.

In another line of research, \emph{recurrence analysis} has been used for software verification \cite{kincaid15,kincaid17}.
Here, the idea is to extract recurrence equations from loops and solve them to summarize the effect of arbitrarily many iterations.
While recurrence analysis often yields very precise summaries for loops without branching, these summaries are conjunctive.
However, for loops with branching, disjunctive summaries are often important to distinguish the branches.

In this work, we embed recurrence analysis into \emph{bounded model checking}
(BMC) \cite{bmc2}, resulting in a robust, competitive model checking algorithm.
To find disjunctive summaries, we exploit the structure of the relational
formula to partition the  state
space on the fly via \emph{model based projection} \cite{spacer} (which
approximates quanti\-fier elimination), and a variation of recurrence analysis
called \emph{transitive projection}.

Our approach is inspired by ABMC \cite{abmc}, which combines BMC with \emph{acceleration} \cite{acceleration-calculus,octagonsP,bozga10}.
ABMC uses \emph{blocking clauses} to speed up the search for counterexamples, but they turned out to be of little use for this purpose.
Instead, they enable ABMC to prove safety of certain challenging benchmarks. This
moti\-vated the development of our novel dedicated algorithm for proving safety via
BMC and blocking clauses.
See \Cref{sec:related} for a detailed comparison with ABMC.

%% file: overview.tex
\subsubsection{Overview}

\paper{\vspace*{-.3cm}}

\begin{algorithm}[t]
  \encode(initial states)\;
  \While{$\top$}{
    \lIf{the current encoding contains an error state}{\Return{$\unknown$}} \nllabel{alg1:fail}
    \encode(one more unrolling of $\to_\tau$ and of all learned relations)\; \nllabel{alg1:unroll}
    \encode(learned relations must not be used twice in a row)\; \nllabel{alg1:trans}
    \lIf{$\neg\checksat()$}{\Return{$\safe$} \nllabel{alg1:safe}}
    $\sigma \gets \getmodel(); \quad \text{ let } \vec{v}_1 \to_{\tau_1} \ldots \to_{\tau_{k-1}} \vec{v}_{k}$ be the run that corresponds to $\sigma$\; \nllabel{alg1:sigma}
    \If{$\vec{v}_s \to_{\tau_s} \ldots \to_{\tau_{s+\ell-1}} \vec{v}_{s+\ell}$ is a loop \nllabel{alg1:loop}}{
      \If{$\vec{v}_s \not\to_{\pi} \vec{v}_{s+\ell}$ for all learned relations $\to_\pi$}{
        learn a transitive relation from the loop and $\sigma$\;\nllabel{alg1:learn}
      }
      pick learned relation $\to_\pi$ with $\vec{v}_s \to_{\pi} \vec{v}_{s+\ell}$\;
      \encode($\to_{\pi}$ has preference over loops from
 $\vec{v}_s$ to $\vec{v}_{s+\ell}$)\;\nllabel{alg1:block}
      backtrack to the point where
$\to_\tau$ was only unrolled $s-1$ times\;
    }
  }
  \caption{TRL (high level); Input: initial \& error states, relation~$\to_\tau$}
  \label{alg:overview}
  \end{algorithm} 
We start with an informal explanation of our approach.
Given a relational formula $\tau$, one can prove safety with BMC by unrolling its transition relation $\to_\tau$
$D$ times, where $D$ is the \emph{diameter} \cite{bmc2}.
So $D$ is the ``longest shortest path'' from an initial to some other state, or more formally:
\[
  D \Def \sup_{\vec{v}' \text{ is reachable from an initial state}} \left( \min \{i \in \NN \mid \vec{v} \text{ is an initial state}, \vec{v} \to_\tau^i \vec{v}' \}\right)
\]
So every reachable state can be reached in $\leq D$ steps.
Hence, if BMC finds no counterexample in $D$ steps, the system is safe.
\Cref{alg:overview} is a high level version of our new algorithm \emph{Transitive Relation Learning} (TRL),
where this observation is exploited in \Cref{alg1:safe}, as unsatisfiability of the underlying SMT problem implies that the diameter has been reached (see the end of this section for more details).
For infinite state systems, $D$ is rarely finite:
Consider the relational formula~\eqref{eq:count} with the initial state $x \doteq 0$.

\smallskip
\noindent
\begin{minipage}{0.49\textwidth}
\begin{equation}
  \label{eq:count}
  x' \doteq x + 1
\end{equation}
\end{minipage}
\begin{minipage}{0.49\textwidth}
\begin{equation}
  \label{eq:learned}
  n > 0 \land x' \doteq x + n
\end{equation}
\end{minipage}

\medskip
\noindent
Then $k \in \NN$ steps are needed to reach a state with $x \doteq k$.
So $D$ is infinite, i.e., the diameter cannot be used to prove safety directly.

The core idea of TRL is to ``enlarge'' $\tau$ to decrease its diameter (even though TRL
never computes the diameter explicitly).
To this end, our approach ``adds'' transitive relations to $\tau$,
which will be called \emph{learned relations} in the sequel (\Cref{alg1:learn}).
For \eqref{eq:count}, TRL would learn the relational formula \eqref{eq:learned}.
Then the diameter of $\eqref{eq:count} \lor \eqref{eq:learned}$ is $1$, as a state with $x \doteq k$ can be reached in $1$ step by setting $n$ to $k$.
Transitive relations are particularly suitable for decreasing the diameter, as they allow us to ignore runs where the same learned relation is used twice in a row.
This is exploited in \Cref{alg1:trans} of \Cref{alg:overview}, where \encode($P$) means that we add an encoding of the property $P$ to the underlying SMT solver.

This raises the questions \emph{when} and \emph{how} new relations should be learned.
Regarding the ``when'', our approach unrolls the  transition relation, just like BMC (\Cref{alg1:unroll}).
Thereby, it looks for \emph{loops} (\Cref{alg1:loop}) and learns a relation when a new loop is encountered (\Cref{alg1:learn}).
However, as we analyze unstructured systems, the definition of a ``loop'' is not obvious.
Details will be discussed in \Cref{sec:loops}.

Regarding the ``how'', TRL ensures that we have a \emph{model} for the loop, i.e., a valuation $\sigma_\Loop$ that corresponds to an evaluation of the loop body (which can be extracted from $\sigma$ in \Cref{alg:overview}).
Then given a loop $\tau_\Loop$ and a model $\sigma_\Loop$, apart from transitivity we only require two more properties for a learned relation $\to_\pi$.
First, the evaluation that corresponds to $\sigma_\Loop$ must also be possible with $\to_\pi$.
This ensures that \Cref{alg:overview} makes ``progress'', i.e., that we do not unroll the same loop with the same model again.
Second, the following set must be finite:
\[
  \{\pi \mid \sigma_\Loop \text{ is a model of } \tau_\Loop, {\to_\pi} \text{ is learned from $\tau_\Loop$ and $\sigma_\Loop$}\}
\]
So TRL only learns finitely many relations from a given loop $\tau_\Loop$.
While TRL may diverge (\Cref{remark:termination}), this ``usually'' ensures termination in practice.

Apart from these restrictions, we have lots of freedom when computing learned relations, as ``enlarging'' $\tau$ (i.e., adding disjuncts to $\tau$) is \emph{always} sound for proving safety.
The \emph{transitive projection} that we use to learn relations (see \Cref{sec:rec}) heavily exploits this freedom.
It modifies the recurrence analysis from \cite{kincaid15} by replacing expensive operations -- convex hulls and polyhedral projections -- by a cheap variation of \emph{model based projection} \cite{spacer}.
While convex hulls and polyhedral projections are over-approximations (for integer arithmetic), model based projections under-approximate.
This is surprising at first, but the justification for using under-approximations is that, as mentioned above, ``enlarging'' $\tau$ is always sound.

Without our modifications, recurrence analysis over-approximates, so we ``mix'' over- and under-approximations.
Thus, learned relations are \emph{not} under-approximations, so \Cref{alg:overview} cannot prove unsafety and returns $\unknown$ in \Cref{alg1:fail}.

Learned relations may reduce the diameter, but computing the diameter is difficult.
Instead, TRL adds \emph{blocking clauses} to the SMT encoding that force the SMT solver to prefer learned relations over loops (\Cref{alg1:block}).
Then unsatisfiability implies that the diameter has been reached, so that $\safe$ can be returned (\Cref{alg1:safe}).

For our example \eqref{eq:count}, once \eqref{eq:learned} has been learned, it is preferred over \eqref{eq:count}.
As \eqref{eq:learned} must not be used twice in a row (\Cref{alg1:trans}), the SMT problem in \Cref{alg1:safe} becomes unsatisfiable after adding a blocking clause for $s = \ell = 1$ (that blocks \eqref{eq:count} for the $1^{st}$ step), and one for $s = 2$ and $\ell = 1$ (that blocks \eqref{eq:count} for the $2^{nd}$ step).
Since we check for reachability of error states after every step (\Cref{alg1:fail}), this implies safety.

\subsubsection*{Outline} After introducing preliminaries in \Cref{sec:preliminaries}, we present our new algorithm
TRL in \Cref{sec:til}.
As TRL builds upon \emph{transitive projections}, we show how to implement such a
projection for linear integer arithmetic in \Cref{sec:rec}.
In \Cref{sec:Unsafety}, we adapt TRL to prove unsafety.
\Cref{sec:related} discusses related work and we evaluate our approach empirically
in \Cref{sec:experiments}. All proofs can be found in \Cref{sec:proofs}.

%% file: preliminaries.tex
\section{Preliminaries}
\label{sec:preliminaries}

We assume familiarity with basics of first-order logic\cite{enderton}.
$\VV$ is a countably infinite set of variables and $\AA$ is a first-order theory with signature $\Sigma$ and carrier $\CC$.
For each entity $e$, $\VV(e)$ is the set of variables that occur in $e$.
$\QF(\Sigma)$ denotes the set of all quantifier-free first-order formulas over $\Sigma$, and $\QF_\land(\Sigma)$ only contains conjunctions of $\Sigma$-literals.
$\top$ and $\bot$ stand for ``true'' and ``false'', respectively.

Given $\psi \in \QF(\Sigma)$ with $\VV(\psi) = \vec{y}$, we say that $\psi$ is $\AA$-\emph{valid} (written $\models_\AA \psi$) if every model of $\AA$ satisfies the universal closure $\forall \vec{y}.\ \psi$ of $\psi$.
A partial function $\sigma: \VV \partial \CC$ is called a \emph{valuation}.
If $\VV(\psi) \subseteq \dom(\sigma)$ and $\models_\AA \sigma(\psi)$, then $\sigma$ is an $\AA$-\emph{model} of $\psi$ (written $\sigma \models_\AA \psi$).
Here, $\sigma(\psi)$ results from $\psi$ by instantiating all variables according to $\sigma$.
If $\psi$ has an $\AA$-model, then $\psi$ is $\AA$-\emph{satisfiable}.
If $\sigma(\psi)$ is $\AA$-satisfiable (but not necessarily $\VV(\psi) \subseteq \dom(\sigma)$), then we say that $\psi$ is $\AA$-\emph{consistent} with $\sigma$.
We write $\psi \models_\AA \psi'$ for $\models_\AA \psi \implies \psi'$, and $\psi \equiv_\AA \psi'$ means $\models_\AA \psi \iff \psi'$.
In the sequel, we omit the subscript $\AA$, and we just say ``valid'', ``model'', ``satisfiable'', and ``consistent''.
We assume that $\AA$ is complete (i.e., $\models \psi$ or $\models \neg \psi$ holds for every closed formula over $\Sigma$) and that $\AA$ has an effective quantifier elimination procedure (i.e., quantifier elimination is computable).

We write $\vec{x}$ for sequences and $x_i$ is the $i^{th}$ element of $\vec{x}$, where $x_1$ denotes the first element.
We use ``$\concat$'' for concatenation of sequences, where we identify sequences of length $1$ with their elements, so, e.g., $x\concat\vec{x} = [x]\concat\vec{x}$.

Let $d \in \NN$ be fixed, and let $\vec{x},\vec{x}' \in \VV^d$ be disjoint vectors of pairwise different variables, called the \emph{pre-} and \emph{post-variables}.
All other variables are \emph{extra variables}.
Each $\tau \in \QF(\Sigma)$ induces a \emph{transition relation} $\to_\tau$ on \emph{states}, i.e., elements of $\CC^d$, where $\vec{v} \to_\tau \vec{v}'$ iff $\tau[\vec{x}/\vec{v},\vec{x}'/\vec{v}']$ is satisfiable.
Here, $[\vec{x}/\vec{v},\vec{x}'/\vec{v}']$ maps
$x^{(\prime)}_i$ to $v^{(\prime)}_i$.

We call $\tau \in \QF(\Sigma)$ a \emph{relational formula} if we are interested in $\tau$'s induced transition relation.
\emph{Transitions} are conjunctive relational formulas without extra variables (i.e., conjunctions of literals over pre- and post-variables).
We sometimes identify $\tau$ with $\to_\tau$, so we may call $\tau$ a relation.

A \emph{$\tau$-run} is a sequence $\vec{v}_1 \to_\tau \ldots \to_\tau \vec{v}_k$.
A \emph{safety problem} $\TT$ is a triple $(\psi_{\init}, \tau,\psi_{\err}) \in \QF(\Sigma) \times \QF(\Sigma) \times \QF(\Sigma)$ where $\VV(\psi_{\init}) \cup \VV(\psi_{\err}) \subseteq \vec{x}$.
 It is \emph{unsafe} if there are $\vec{v},\vec{v}' \in \CC^d$ such that  $[\vec{x} / \vec{v}] \models \psi_\init$, $\vec{v} \to^*_\tau \vec{v}'$, and $[\vec{x} / \vec{v}'] \models \psi_\err$.

Throughout the paper, we use $c,d,e,k,\ell,s$ for integer constants (where $d$ always
denotes the size of $\vec{x}$, and $s$ and $\ell$ always denote the start and length of a
loop\report{ as in \Cref{alg1:loop} of \Cref{alg:overview}}); $\vec{v}$ for states;
$w,x,y$ for 
variables; $\tau,\pi$ for relational formulas; $\sigma,\theta$ for valuations; and $\mu$
for variable renamings.

%% file: til.tex
\begin{algorithm}[t]
  $b \gets 0; \quad \vec{\pi} \gets [\tau]; \quad \blocked \gets \emptyset$\; \nllabel{alg:init}
  $\add(\mu_{1}(\psi_\init))$ \tcp*{encode the initial states}
  \While(\tcp*[f]{main loop}){$\top$}{
    $b\increment; \quad \push(); \quad \add(\mu_{b}(\psi_\err))$ \tcp*{encode the error states} \nllabel{alg:err1}
    \leIf{$\checksat()$}{
      \Return{$\unknown$} \nllabel{alg:err2}
    }{
      $\pop()$ \tcp*[f]{check their reachability\hspace{-.9em}}
    }
    $\push()$ \tcp*{add backtracking point}
    \lIf(\tcp*[f]{encode transitivity}){$b>1$}{$\add(\ind[\id]{x}{b} \doteq 1 \lor
      \ind[\id]{x}{b} \not\doteq \ind[\id]{x}{b-1})$} \nllabel{alg:trans}
    $\add(\mu_{b}(\bigvee_{n=1}^{|\vec{\pi}|} (\pi_n \land x_\id \doteq n)))$ \tcp*{encode
      $\to_\tau$ and learned relations} \nllabel{alg:unroll}
    $\add(\bigwedge_{(b,\pi) \in \blocked} \pi)$ \tcp*{add blocking clauses for this $b$} \nllabel{alg:block1}
    \lIf{$\neg\checksat()$}{
      \Return{$\safe$} \nllabel{alg:safe} \tcp*[f]{check if the search space is exhausted}}
    $\sigma \gets \getmodel(); \quad \vec{\tau} \gets \trace_b(\sigma,\vec{\pi})$ \nllabel{alg:trace} \tcp*{build trace from current model}
    \If(\tcp*[f]{search loop}){$[\tau_s,\ldots,\tau_{s+\ell-1}]$ is a loop \nllabel{alg:loop}}{
      $\sigma_{\Loop} \gets [x / \sigma(\mu_{s,\ell}(x)) \mid x \in \vec{x} \cup \vec{x}']$ \nllabel{alg:model}\tcp*{build the valuation for the loop}
      \If(\tcp*[f]{redundancy check}){no $\pi \in \tail(\vec{\pi})$ is consistent with $\sigma_{\Loop}$ \nllabel{alg:redundant}}{
        $\tau_{\Loop} \gets \mu_{s,\ell}^{-1}(\bigwedge_{i=s}^{s+\ell-1} \mu_{i}(\tau_i))$ \nllabel{alg:learn1} \tcp*{build the loop}
        $\vec{\pi} \gets \vec{\pi}\concat\tip(\tau_{\Loop}, \sigma \circ \mu_{s,\ell})$ \nllabel{alg:learn2} \tcp*[f]{learn relation}
      }
      $\text{let } \pi \in \tail(\vec{\pi}) \text{ and } \overline{\sigma} \supseteq
      \sigma_{\Loop}$ s.t.\ $\overline{\sigma} \models
\pi$ \tcp*{pick suitable learned relation} \nllabel{alg:pick}
      $\blocked.\add(s+\ell-1,\blockingclause(s,\ell,\pi,\overline{\sigma}))$ \tcp*{block the loop} \nllabel{alg:block2}
      \lWhile{$b \geq s$}{$\{ \, \pop(); \ b\decrement \,\}$ \tcp*[f]{backtrack to the start
          of the loop}\nllabel{alg:backtrack}}
    }
  }
  \caption{TRL -- Input: a safety problem $\TT = (\psi_\init,\tau,\psi_\err)$}
  \label{alg}
\end{algorithm}

\section{Transitive Relation Learning}
\label{sec:til}
In this section, we present our novel model checking algorithm \emph{Transitive Relation
Learning} (TRL) in detail,  see \Cref{alg}.
Here, for all
$i,j \in \NN_+ = \NN \setminus \{0\}$
we define $\mu_{i,j}(x') \Def \ind{x}{i+j}$ if $x' \in \vec{x}'$ and $\mu_{i,j}(x) \Def \ind{x}{i}$, otherwise.
So in particular, we have $\mu_{i,j}(\vec{x}) = \ind{\vec{x}}{i}$ and $\mu_{i,j}(\vec{x}') = \ind{\vec{x}}{i+j}$, where we assume that $\ind{\vec{x}}{1},\ind{\vec{x}}{2}, \ldots \in \VV^d$ are disjoint vectors of pairwise different fresh variables.
Intuitively, the variables $\ind{\vec{x}}{i}$ represent the $i^{th}$ state in a run, and applying $\mu_{i,j}$ to a relational formula yields a formula that relates the $i^{th}$ and the $(i+j)^{th}$ state of a run.
For convenience, we define $\mu_{i} \Def \mu_{i,1}$ for all $i \in \NN_+$, i.e., $\mu_i(\vec{x}) = \ind{\vec{x}}{i}$ and $\mu_i(\vec{x}') = \ind{\vec{x}}{i+1}$.
As in SMT-based BMC, TRL uses an incremental SMT solver to unroll the transition relation step by step (\Cref{alg:unroll}), but in contrast to BMC, TRL infers \emph{learned relations} on the fly (\Cref{alg:learn2}).
The \emph{input formula} $\tau$ as well as all learned relations are stored in $\vec{\pi}$.
Before each unrolling, we set a backtracking point with the command $\push$ and add a suitably variable-renamed version of the description of the error states to the SMT problem (i.e., to the state of the underlying SMT solver) in \Cref{alg:err1}.
Then the command $\checksat$ checks for reachability of error states, and the command $\pop$ removes all formulas from the SMT problem that have been added since the last invocation of $\push$ (\Cref{alg:err2}), i.e., it removes the encoding of the error states (unless the check succeeds, so that TRL fails).
For each unrolling, suitably variable renamed variants of $\vec{\pi}$'s elements are added to the underlying SMT problem with the command $\add$ in \Cref{alg:unroll}.
If no error state is reachable after $b-1$ steps, but the transition relation cannot be
unrolled $b$ times (i.e., the SMT problem that corresponds to the $b$-fold unrolling is
unsatisfiable), then the diameter of the analyzed system (including learned relations) has
been reached, and hence safety has been proven (\Cref{alg:safe}).

The remainder of this section is structured as follows:
First, \Cref{sec:basics} introduces \emph{conjunctive variable projections} that
 are used to
compute the \emph{trace} (\Cref{alg:trace} of \Cref{alg}). 
Next, \Cref{sec:loops} defines \emph{loops} and
discusses how to find \emph{non-redundant loops} that are suitable for learning new relations (\Cref{alg:loop,alg:model,alg:redundant,alg:learn1}).
Then, \Cref{sec:transitiveProjections} introduces \emph{transitive projections}
that are used to learn relations (\Cref{alg:trans,alg:learn2}).
\report{Afterwards}\paper{Finally}, \Cref{sec:block} presents
\emph{blocking clauses}, which ensure that learned re\-la\-tions are preferred over other (sequences of) transitions
(\Cref{alg:block1,alg:pick,alg:block2,alg:backtrack}).
\report{Finally, we illustrate \Cref{alg} with
a complete example in \Cref{sec:example}.}

\subsection{Conjunctive Variable Projections and Traces}
\label{sec:basics}

To decide when to learn a new relation, TRL inspects the \emph{trace} (Lines \ref{alg:trace} and \ref{alg:loop}).
The trace is a sequence of transitions induced by the formulas that have been added to the SMT problem while unrolling the transition relation, and by the current model (\Cref{alg:trace}).
To compute them,  we use \emph{conjunctive variable projections}, which are like
\emph{model based projections} \cite{spacer}, but always 
yield conjunctions.%
\begin{definition}[Conjunctive Variable Projection]
  \label{def:projections}
  A function
  \[
    \mbip: \QF(\Sigma) \times (\VV \partial \CC) \times 2^{\VV} \to \QF(\Sigma)
  \]
  is called a
  \emph{conjunctive variable projection} if 

  \vspace{-0.3em}
  \noindent
  \begin{minipage}[t]{0.49\textwidth}
    \begin{enumerate}
    \item $\sigma \models \mbip(\tau,\sigma,X)$
    \item $\mbip(\tau,\sigma,X) \models \tau$
    \item $\{\mbip(\tau,\theta,X) \mid \theta \models \tau\}$ is finite
    \end{enumerate}
  \end{minipage}
  \begin{minipage}[t]{0.49\textwidth}
    \begin{enumerate}
      \setcounter{enumi}{3}
    \item $\VV(\mbip(\tau,\sigma,X)) \subseteq X \cap \VV(\tau)$
    \item $\mbip(\tau,\sigma,X) \in \QF_\land(\Sigma)$
    \end{enumerate}
  \end{minipage}
  \medskip

  \noindent
  for all $\tau \in \QF(\Sigma)$, $X \subseteq \VV$, and $\sigma \models
  \tau$.
  We abbreviate $\mbip(\tau,\sigma,\vec{x} \cup \vec{x}')$ by $\mbip(\tau,\sigma)$.
\end{definition}
So like model based projection, $\mbip$ under-approximates quantifier elimination by projecting to the variables $X$ (by (2) and (4)).
To do so, it implicitly performs a finite case analysis (by (3)), which is driven by
the model $\sigma$ (by (1)).
In contrast to model based projections, $\mbip$ always yields conjunctions (by (5)).
Note that by (1),
$\mbip(\tau,\sigma,X)$ may only contain variables from $\dom(\sigma)$.

\begin{remark}[$\mbip$ and $\mathsf{mbp}$]
  Conjunctive variable projections are obtained by combining a model based projection $\mathsf{mbp}$ (which satisfies \Cref{def:projections} (1--4)) with \emph{syntactic implicant projection} $\sip$ \cite{adcl}.
  To see how to compute $\mathsf{mbp}$ for linear integer arithmetic, recall that Cooper's
  method for quantifier elimination \cite{cooper72} essentially maps
    $\exists \vec{y}.\ \tau$ to a disjunction of formulas of the form
  $\tau[\vec{y}/\vec{t}]$, where the variables $\vec{y}$ do not occur in the terms $\vec{t}$.
  Instead, $\mathsf{mbp}$ just computes one of these disjuncts, which is satisfied by the provided model.
  For $\sip(\tau,\sigma)$, one computes the conjunction of all literals of $\tau$'s negation normal form that are satisfied by $\sigma$.
  Then $\mbip(\tau,\sigma) \Def \sip(\mathsf{mbp}(\tau,\sigma),\sigma)$.
\end{remark}

\begin{remark}[$\mbip$ and Quantifier Elimination]
  \Cref{def:projections} (1--4) imply
  \begin{align*}
    \exists \vec{y}.\ \tau & {} \mbox{$\; \equiv \; \bigvee_{\sigma \models \tau} \mbp(\tau,\sigma)$}
     \label{mbp-property} \qquad \text{where $\vec{y}$ are $\tau$'s extra variables.}
\end{align*}
  So $\mbp$ yields
  a quantifier elimination procedure $\mathsf{qe}$ which maps $\exists
\vec{y}.\ \tau$ to $\mathit{res}$:

\vspace{-.9em}

\[\mbox{\small $\mathit{res} \gets \bot;\hspace{.75em}$ \textbf{while}
  $\tau$ \emph{has a model}
  $\sigma\;$ $\{\mathit{res} \gets \mathit{res} \lor \mbp(\tau,\sigma);\hspace{.75em} \tau
    \gets \tau \land \neg \mbp(\tau,\sigma)\}$} 
  \]

\noindent
But for a single model $\sigma$, $\mbp(\tau,\sigma)$ under-ap\-prox\-i\-mates quantifier elimination.
\end{remark}
So $\mbip(\tau,\sigma)$ just computes one disjunct of $\mathsf{qe}(\exists \vec{y}.\ \tau)$ which is satisfied by $\sigma$.
However, like model based projection, $\mbip$ can be implemented efficiently for many theories with effective, but very expensive quantifier elimination procedures.%
\begin{example}[$\mbip$]
  \label{ex:projections}
  Consider the following formula $\tilde{\tau}$:
  \small
  \[
    \begin{array}{rcl}
     ( (w \doteq 0 \land \ind{x}{2} \doteq x\!+\!1 \land \ind{y}{2} \doteq y\!+\!1)  &
      \lor &  (\ind{w}{2} \doteq w \land w \doteq 1 \land \ind{x}{2} \doteq x\!-\!1 \land
      \ind{y}{2} \doteq y\!-\!1) ) \land {} \\
      ( (\ind{w}{2} \doteq 0 \land x' \doteq \ind{x}{2}\!+\!1 \land y' \doteq
      \ind{y}{2}\!+\!1)  & \lor &  (w' \doteq \ind{w}{2} \land \ind{w}{2} \doteq 1 \land x'
      \doteq \ind{x}{2}\!-\!1 \land y' \doteq \ind{y}{2}\!-\!1) )
    \end{array}
  \]
  \normalsize It encodes two steps with \Cref{ex:ex1},
  where $\ind{\vec{x}}{2} = [\ind{w}{2},\ind{x}{2},\ind{y}{2}]$ represents the values after one
  step.
  In \Cref{alg:trace}, \Cref{alg} might find a run like $\sigma(\ind{\vec{x}}{1}) \to_\tau \sigma(\ind{\vec{x}}{2}) \to_\tau \sigma(\ind{\vec{x}}{3})$ for
  \[
  \begin{array}{rcl@{\;\;}l@{\;\;}l}
  \sigma & \Def  & [\ind{w}{1}/\ind{x}{1}/\ind{y}{1}/0, & \ind{w}{2}/\ind{x}{2}/\ind{y}{2}/1, & \ind{w}{3}/1,\ind{x}{3}/\ind{y}{3}/0].
  \end{array}
  \]
 Here, $[w/x/y/c, \ldots]$ abbreviates $[w/c, x/c, y/c, \ldots]$.
  Then the variable renaming $\mu_{1,2}$ allows us to
instantiate the pre- and post-variables
by the first and last state,
resulting in the following model of
$\tilde{\tau}$:
   \[
    \tilde{\sigma} \Def \sigma \circ \mu_{1,2} = \sigma \cup [w/x/y/0,\;\; w'/1,x'/y'/0] \quad \text{where } (\sigma \circ \mu_{1,2})(x) = \sigma(\mu_{1,2}(x))
  \]
  To get rid of $\ind{w}{2},\ind{x}{2},\ind{y}{2}$, one could compute $\mathsf{qe}(\exists \ind{w}{2},\ind{x}{2},\ind{y}{2}.\ \tilde{\tau})$, resulting in:
  \begin{align}
    & w \doteq 0 \land x' \doteq x+2 \land y' \doteq y+2 \tag{\ensuremath{\inc}} \\
    {} \lor {} & w \doteq 0 \land w' \doteq 1 \land x' \doteq x \land y' \doteq y \label{eq:mbp-ex} \tag{\ensuremath{\mathsf{eq}}}\\
    {} \lor {} &w \doteq 1 \land w' \doteq 1 \land x' \doteq x-2 \land y' \doteq y-2. \tag{\ensuremath{\dec}}
  \end{align}
  Instead, we may have $\mbp(\tilde{\tau},\tilde{\sigma}) = \eqref{eq:mbp-ex}$, as
  $\tilde{\sigma} \models \eqref{eq:mbp-ex}$.
\end{example}
Intuitively, a relational formula $\tau$ describes how states can change, so it is composed of many different cases.
These cases may be given explicitly (by disjunctions) or implicitly (by
extra variables, which express non-determinism).
Given a model $\sigma$ of $\tau$ that describes a \emph{concrete} change of state, $\mbip$ computes a description of the corresponding case.
Computing \emph{all} cases amounts to eliminating all extra variables and converting the result to DNF, which is impractical.

When unrolling the transition relation in \Cref{alg:unroll} of \Cref{alg}, we identify
each relational formula $\pi_n$ with its index $n$ in the sequence $\vec{\pi}$.
To this end, we use a fresh variable $x_\id$, and our SMT encoding forces
$\ind[\id]{x}{i}$ to be the identifier of the relation that is used for the $i^{th}$ step.
Similarly to \cite{abmc}, the \emph{trace} is the sequence of transitions that results from applying $\mbip$ to the unrolling
of the transition relation that is constructed by \Cref{alg} in \Cref{alg:unroll}.
So a trace is a sequence of transitions that can be applied subsequently, starting in an initial state.
\begin{definition}[Trace]
  \label{def:trace}
  Let $\vec{\pi}$ be a sequence of relational formulas, let
  \begin{align}
    \label{eq:trace}
    \paper{\textstyle}
    \sigma \models \bigwedge_{i=1}^{b} \mu_{i}\left(\bigvee_{n=1}^{|\vec{\pi}|} (\pi_n \land x_\id \doteq n)\right) \qquad \text{where $b \in \NN_+$},
  \end{align}
  and let $\id(i) \Def \sigma(\ind[\id]{x}{i})$.
  Then the \emph{trace induced by $\sigma$} is
  \[
    \trace_b(\sigma,\vec{\pi}) \Def [\mbip(\pi_{\id(i)}, \sigma \circ \mu_{i})]_{i=1}^{b}.
  \]
\end{definition}

Recall that $\mu_i$ renames $\vec{x}$ and $\vec{x}'$ into $\ind{\vec{x}}{i}$ and
$\ind{\vec{x}}{i+1}$, and $\id(i) = \sigma(\ind[\id]{x}{i})$ is the
index
of the
relation from $\vec{\pi}$ that is used for the $i^{th}$ step.
So each model $\sigma$ of \eqref{eq:trace} corresponds to a run
$\sigma(\mu_1(\vec{x}))
\to_{\pi_{\id(1)}} \ldots \to_{\pi_{\id(b-1)}} \sigma(\mu_{b}(\vec{x})) \to_{\pi_{\id(b)}} \sigma(\mu_{b}(\vec{x}'))$, and the trace
induced by $\sigma$ contains the transitions that were used in this run.

\begin{example}[Trace]
  \label{ex:trace}
  Consider the extension of $\sigma$ from \Cref{ex:projections} with $[\ind[\id]{x}{1}/1,\;
    \ind[\id]{x}{2}/1]$:
  \[
  \sigma \Def [\ind{w}{1}/\ind{x}{1}/\ind{y}{1}/0,\ind[\id]{x}{1}/1, \quad
    \ind{w}{2}/\ind{x}{2}/\ind{y}{2}/\ind[\id]{x}{2}/1, \quad \ind{w}{3}/1,\ind{x}{3}/\ind{y}{3}/0]
  \]
  Thus, $\id(1) = \sigma(\ind[\id]{x}{1}) = 1$, $\id(2) =  \sigma(\ind[\id]{x}{2}) = 1$, and
  $\pi_{\id(1)} = \pi_{\id(2)} =
  \pi_1 = \tau$.
  Then
  \begin{align*}
    & \trace_2(\sigma, [\tau,\tau]) = [\mbip(\tau,\sigma \circ \mu_{1}), \mbip(\tau,\sigma \circ \mu_{2})]           \\
    {} = {} & [\mbip(\tau,[w/x/y/0, \; w'/x'/y'/1]), \mbip(\tau,[w/x/y/1, \; w'/1,x'/y'/0])] = [\tau_\inc, \tau_\dec].
  \end{align*}
\end{example}

\subsection{Loops}
\label{sec:loops}

As $\vec{\pi}$ only gives rise to finitely many transitions, the trace is bound to contain \emph{loops}, eventually (unless \Cref{alg} terminates beforehand).
\begin{definition}[Loop]
  A sequence of transitions $\tau_1,\ldots,\tau_k$ is called a \emph{loop} if there are $\vec{v}_0,\ldots,\vec{v}_{k+1} \in \CC^d$ such that $\vec{v}_0 \to_{\tau_1} \ldots \to_{\tau_k} \vec{v}_{k} \to_{\tau_1} \vec{v}_{k + 1}$.
\end{definition}
Intuitively, these loops are the reason why BMC may diverge.
To prevent divergence, TRL learns a new relation when a loop is detected (\Cref{alg:loop}).
\begin{remark}[Finding Loops]
  \label{remark:finding-loops}
  Loops can be detected by SMT solving.
A cheaper way is to look for duplicates on the trace, but then loops are found ``later'', as
a trace $[\ldots, \pi, \pi, \ldots]$ is needed to detect a loop $\pi$, but one occurrence of $\pi$ is insufficient.
As a trade-off between precision and efficiency, our im\-ple\-men\-ta\-tion uses an approximation based on \emph{dependency graphs} \cite{abmc}.
More precisely, our implementation maintains a graph $\mathcal{G}$ whose nodes are
transitions, and it adds an edge between two transitions if they occur subsequently on the
trace at some point.
Then it considers a sub\-sequence $\tau_i,\ldots,\tau_j$ of the trace to be a potential loop if
$\mathcal{G}$ contains an edge from $\tau_j$ to $\tau_i$.
\end{remark}

\begin{remark}[Disregarding ``Learned'' Loops]
  \label{remark:loops}
    One should disregard ``loops'' consisting of a single \emph{learned transition}, i.e.,
    a transition that results from applying $\mbip$ to some $\pi \in \tail(\vec{\pi})$,
    where $\tail(\tau\concat\vec{\pi}') \Def \vec{\pi}'$.
    Here, $\tail(\vec{\pi})$ contains all learned relations, as the first element of $\vec{\pi}$ is the input formula $\tau$.
    The reason is that our goal is to deduce transitive relations, but learned relations are already transitive.
    In the sequel, we assume that the check in \Cref{alg:loop} fails for such loops.
\end{remark}
If there are several choices for $s$ and $\ell$ in \Cref{alg:loop}, then our implementation only considers loops of minimal length and, among those, it minimizes $s$.
\begin{example}[Detecting Loops]
  \label{ex:loops}
  Consider the following model for $\tau$ from \Cref{ex:ex1}.
  \[
    \sigma \Def [\ind{w}{1}/\ind{x}{1}/\ind{y}{1}/0, \ind[\id]{x}{1}/1, \quad \ind{w}{2}/0,\ind{x}{2}/\ind{y}{2}/1]
  \]
   Then $\trace_1(\sigma, [\tau]) = [\tau_\inc]$.
   As $\tau_\inc$ is a loop,\footnote{Depending on the technique that is used to detect
   loops, an actual implementation might require one more unrolling of $\to_{\tau}$
   to obtain the trace $[\tau_\inc,\tau_\inc]$ in order to detect the loop
   $\tau_\inc$, see \Cref{remark:finding-loops}.}
     this causes TRL to learn a relation like the following one
 (see \Cref{sec:transitiveProjections} for details).
 \begin{align*}
  w \doteq 0 \land x' > x \land x' - x \doteq y' - y
  \tag{$\tau^+_{\inc}$}
  \end{align*}
\end{example}
TRL only learns relations from loops that are \emph{non-redundant} w.r.t.\ all relations that have been learned before \cite{adcl}.
\begin{definition}[Redundancy]
  \label{def:redundancy}
  If ${\to_\tau} \subseteq {\to_{\tau'}}$, then $\tau$ is \emph{redundant} w.r.t.\ $\tau'$.
\end{definition}
\begin{example}
  The relation $\tau_\inc$ is redundant w.r.t.\ $\tau^+_\inc$, but $\tau_\dec$ is not.
\end{example}
\Cref{alg:redundant} uses a sufficient criterion for non-redundancy:
If all learned relations are falsified by the values before and after the loop, then
$\tau_s, \ldots, \tau_{s + \ell -1}$ cannot be simulated by a previously learned relation, so it is non-redundant and we learn a new relation.
The values before and after the loop are obtained from the current model $\sigma$ by
setting $\vec{x}$ to $\sigma(\ind{\vec{x}}{s})$ and $\vec{x}'$ to
$\sigma(\ind{\vec{x}}{s+\ell})$, i.e., we use $\sigma \circ \mu_{s,\ell}$ in
\Cref{alg:model}.

To learn a new relation,
we first compute the  relation
\begin{equation}
  \label{eq:loop-formula}
  \paper{\textstyle}
  \tau_\Loop \Def \mu_{s,\ell}^{-1}(\phi_\Loop) \qquad \text{where} \qquad \phi_\Loop \Def
  \bigwedge_{i=s}^{s+\ell-1} \mu_{i}(\tau_i) 
\end{equation}
of the loop in \Cref{alg:learn1}, where $\mu^{-1}_{s,\ell}$ is the inverse of $\mu_{s,\ell}$.
So in \Cref{ex:loops}, we have $\sigma \circ \mu_{1,1} \supseteq [w/x/y/0,\,
  w'/0,x'/y'/1]$ and $\tau_\Loop \Def \mu^{-1}_{1,1}(\mu_1(\tau_\inc)) = \tau_\inc$ as $s = \ell = 1$.
So $\sigma \circ \mu_{s,\ell}$ indeed corresponds to one evaluation of the loop, as $\sigma \circ \mu_{s,\ell} \models \tau_\Loop$.

To see that $\tau_\Loop$ is also the desired relation in general, note that $\phi_\Loop$ is the conjunction of the transitions that constitute the loop, where all variables are renamed as in \Cref{alg:unroll} of \Cref{alg}, i.e., in such a way that the post-variables of the $i^{th}$ step are equal to the pre-variables of the $(i+1)^{th}$ step.
So we have $\sigma \models \phi_\Loop$ and thus $\sigma \circ \mu_{s,\ell} \models \tau_\Loop$.
Hence,
we can use $\tau_\Loop$ and $\sigma \circ \mu_{s,\ell}$ to learn a new relation via
so-called \emph{transitive projections}
in \Cref{alg:learn2}.

\subsection{Transitive Projections}
\label{sec:transitiveProjections}

We now define \emph{transitive projections} that approximate transitive closures of loops.
As explained in \Cref{sec:overview}, we do not restrict ourselves to under- or over-approximations, but we allow ``mixtures'' of both.
Analogously to $\mbip$, transitive projections perform a finite case analysis that is
driven by the provided model $\sigma$.

\begin{definition}[Transitive Projection]
  \label{def:ti}
  A function
  \[
    \tip:  \QF(\Sigma) \times (\VV \partial \CC)  \to \QF(\Sigma)
  \]
  is called a \emph{transitive projection}
  if the following holds for all transitions $\tau \in \QF(\Sigma)$ and all $\sigma
  \models \tau$:

  \vspace*{-.15cm}

  \noindent  
  \begin{minipage}[t]{0.49\textwidth}
    \begin{enumerate}
    \item $\tip(\tau,\sigma)$ is consistent with $\sigma$
    \item $\{\tip(\tau,\theta) \mid \theta \models \tau\}$ is finite
    \end{enumerate}
  \end{minipage}
  \begin{minipage}[t]{0.49\textwidth}
    \begin{enumerate}
      \setcounter{enumi}{2}
    \item $\to_{\tip(\tau,\sigma)}$ is transitive
    \end{enumerate}
  \end{minipage} 
\end{definition}

Clearly, the specifics of $\tip$ depend on the underlying theory.
Our implementation of $\tip$ for quantifier-free linear integer arithmetic is explained in \Cref{sec:rec}.

\begin{example}[$\tip$]
  \label{ex:Transition Invariants}
  For \Cref{ex:ex1}, $\tau_\ti \Def x' - x \doteq y' - y$ over-approximates the transitive
  closure $\to^+_\tau$.
  Such over-approximations are also called \emph{transition invariants} \cite{transition_invariants}.
  With $\tau_\ti$, one can prove safety for any $\psi_\init$ with $\psi_\init \models x \doteq y$, as then $\psi_\init \land \tau_\ti \models x' \doteq y'$, which shows that no error state with $w \doteq 1 \land x \leq 0 \land y > 0$ is reachable.

  By using $\mbip$,
  TRL instead considers $\tau_\inc$ and $\tau_\dec$ separately and learns
  \begin{align*}
    \tip(\tau_\inc,\sigma_\inc) & {} \Def w \doteq 0 \land x' > x \land x' - x \doteq y' - y                   \tag{$\tau^+_{\inc}$} \\
    \tip(\tau_\dec,\sigma_\dec) & {} \Def w' \doteq w \land w \doteq 1 \land x' < x \land x' - x \doteq y' - y \tag{$\tau^+_{\dec}$}
  \end{align*}
  if $\sigma_\inc \models \tau_\inc$ and $\sigma_\dec \models \tau_\dec$.
  In this way, \Cref{alg} can learn disjunctive relations like $\tau^+_{\inc} \lor \tau^+_{\dec}$, even if $\tip$ only yields conjunctive relational formulas (which is true for our current implementation of $\tip$ -- see \Cref{sec:rec} -- but not enforced by \Cref{def:ti}).
\end{example}
In contrast to conjunctive variable projections, $\tip(\tau,\sigma)$ may contain extra variables that do not occur in $\tau$ (which will be exploited in \Cref{sec:rec}).
Hence, instead of $\sigma \models \tip(\tau,\sigma)$ we require consistency with $\sigma$, i.e., $\sigma(\tip(\tau,\sigma))$ must be satisfiable.

\begin{remark}[Properties of $\tip$]
  \label{remark:properties-tip}
 Due to \Cref{def:ti} (1), our definition of $\tip$ implies
 \[
  \paper{\textstyle}
  \tau \models \exists \vec{y}. \, \bigvee_{\sigma \models \tau} \tip(\tau,\sigma), \quad \text{and thus,}
  \quad
  {\to_{\tau}} \subseteq \bigcup_{\sigma \models \tau}
  {\to_{\tip(\tau,\sigma)}},
\]
where $\vec{y}$ are the extra variables of $\bigvee_{\sigma \models \tau} \tip(\tau,\sigma)$.
However, \Cref{def:ti} does \emph{not} ensure
${\to^+_\tau} \subseteq \bigcup_{\sigma \models \tau}
  {\to_{\tip(\tau,\sigma)}}$.
  So there is no guarantee that $\tip$ covers $\to^+_\tau$ entirely, i.e., $\tip$ cannot be used to compute transition invariants, in general.
  \Cref{def:ti} does not ensure ${\to^+_\tau} \supseteq \bigcup_{\sigma \models \tau}
  {\to_{\tip(\tau,\sigma)}}$ either, as $\tip(\tau,\sigma)$ does not imply $\sigma(\vec{x}) \to^+_\tau \sigma(\vec{x}')$.
\end{remark}

\begin{example}
  \label{Counterex-tip}
  To see that $\tip$ computes no over- or under-approximations, let
  \[
    \tau \Def x' \doteq x + 1 \land y' \doteq y + x.
  \]
  Then for all $\sigma \models \tau$, we might have:
  \[
    \tip(\tau,\sigma) =
    \begin{cases}
      x \geq 0 \land x' > x \land y' \geq y, & \text{if } \sigma(x) \geq 0 \\
      x < 0 \land x' > x \land y' < y,   & \text{if } \sigma(x) < 0
    \end{cases}
  \]
  However, $(x \geq 0 \land x' > x \land y' \geq y) \lor (x < 0 \land x' > x \land y' < y)$ is not an over-approximation of $\to^+_\tau$ (i.e., ${\to^+_\tau} \not\subseteq \bigcup_{\sigma \models \tau}
    {\to_{\tip(\tau,\sigma)}}$), as we have, e.g.,
  \[
    (-1,0) \to_\tau (0,-1) \to_\tau (1,-1) \to_\tau (2,0), \qquad \text{but} \qquad (-1,0) \not\to_{\tip(\tau,\sigma)} (2,0)
  \]
  for all $\sigma \models \tau$.
  Moreover, we also have ${\to^+_\tau} \not\supseteq \bigcup_{\sigma \models \tau} {\to_{\tip(\tau,\sigma)}}$, since
  \[
    (-1,0) \to_{\tip(\tau,\sigma)} (10,-20), \qquad \text{but} \qquad (-1,0)
  \not\to^+_\tau (10,-20)\paper{ \pagebreak[3]} 
  \]
  if $\sigma(x) < 0$.
  In contrast to $\tip(\tau, \sigma)$, linear over-approximations for $\to^+_\tau$ like $x' > x$
  cannot distinguish whether $y$ increases or decreases.
\end{example}

As TRL proves safety via \emph{blocking clauses} (\Cref{sec:block}) that only block steps that are cov\-er\-ed by learned relations, the fact that $\tip$ does not yield over-ap\-prox\-i\-ma\-tions does not affect soundness.
However, it may cause divergence (\Cref{remark:termination}).

Recall that our SMT encoding forces $\ind[\id]{x}{i}$ to be the identifier of the relation that is used for the $i^{th}$ step (\Cref{alg:unroll}).
To exploit transitivity of $\tip$, we add the constraint
$\ind[\id]{x}{b} \doteq 1 \lor \ind[\id]{x}{b} \not\doteq \ind[\id]{x}{b-1}$ in
\Cref{alg:trans}, so that learned relations (with index $>1$)\linebreak are not used several times in a row, since this is unnecessary for transitive relations.

\subsection{Blocking Clauses}
\label{sec:block}

In \Cref{alg:pick}, we are guaranteed to find a learned relation $\pi$ which is consistent
with $\sigma_\Loop$: If our sufficient criterion for non-redundancy in
\Cref{alg:redundant} failed, then the existence of $\pi$ is guaranteed.
Otherwise, we learned a new relation $\pi$ in \Cref{alg:learn2} which is consistent with $\sigma_\Loop \subseteq \sigma \circ \mu_{s,\ell}$ by definition of $\tip$.
Thus, we can use $\pi$ and a model $\overline{\sigma} \supseteq \sigma_\Loop$ of $\pi$ to record a \emph{blocking clause} in \Cref{alg:block2}.
\begin{definition}[Blocking Clauses]
  \label{def:blocking}
  Consider a relational formula $\pi$, and let $\overline{\sigma}$ be a model of $\pi$.
  We define:
  \[
    \blockingclause(s,\ell,\pi,\overline{\sigma}) \Def
    \begin{cases}
      \mu_{s,\ell}(\neg \mbip(\pi, \overline{\sigma})) \lor \ind[\id]{x}{s} > 1, & \text{if } \ell = 1 \\
      \mu_{s,\ell}(\neg \mbip(\pi, \overline{\sigma})),                         & \text{if } \ell > 1
    \end{cases}
  \]
\end{definition}
Here, $s$ and $\ell$ are natural numbers such that $[\tau_i]_{i=s}^{s+\ell-1}$ is a (possibly) redundant loop on the trace.
Blocking clauses exclude models that correspond to runs
\begin{equation}
  \label{blockedRun}
  \vec{v}_1 \to_{\tau_1} \ldots \to_{\tau_{s-1}} \vec{v}_s \to_{\tau_{s}} \ldots \to_{\tau_{s+\ell-1}} \vec{v}_{s+\ell}
\end{equation}
where $\vec{v}_{s} \to_{\pi} \vec{v}_{s+\ell}$. Intuitively, if $\ell = 1$ then
$\blockingclause(s,\ell,\pi,\overline{\sigma})$ states that one may still evaluate 
$\vec{v}_s$ to
$\vec{v}_{s+\ell}$, but one has to use a learned transition. If $\ell > 1$, then
$\blockingclause(s,\ell,\pi,\overline{\sigma})$
states that one may still  evaluate 
$\vec{v}_s$ to
$\vec{v}_{s+\ell}$, but not in $\ell$ steps.
More precisely, 
blocking clauses take into account that%

\vspace{-0.5em}
\noindent
\begin{minipage}{0.44\textwidth}
  \begin{equation}
    \label{prefix}
    \vec{v}_1 \to_{\tau_1} \ldots \to_{\tau_{s+\ell-2}} \vec{v}_{s+\ell-1}
  \end{equation}
\end{minipage}
\begin{minipage}{0.1\textwidth}
  \begin{equation*}
    \text{and}
  \end{equation*}
\end{minipage}
\begin{minipage}{0.44\textwidth}
  \begin{equation}
    \label{unblockedRun}
    \vec{v}_1 \to_{\tau_1} \ldots \to_{\tau_{s-1}} \vec{v}_s \to_{\pi} \vec{v}_{s+\ell}
  \end{equation}
\end{minipage}

\medskip
\noindent
must not be blocked to ensure that $\vec{v}_2,\ldots,\vec{v}_{s+\ell}$ remain reachable.
For the former, note that blocking clauses affect the suffix $\vec{v}_{s} \to_{\tau_s} \ldots \to_{\tau_{s+\ell-1}} \vec{v}_{s+\ell}$ of \eqref{blockedRun} (as they contain $\mu_{s,\ell}(\neg \mbip(\pi, \overline{\sigma}))$), but not \eqref{prefix}, so $\vec{v}_{2},\ldots,\vec{v}_{s+\ell-1}$ remain reachable.

Regarding \eqref{unblockedRun}, note that \eqref{blockedRun} corresponds to one unrolling of the loop (without using the newly learned
relation $\pi$).
In contrast, \eqref{unblockedRun} simulates one unrolling of the loop using the new relation $\pi$.
To see that our blocking clauses prevent the sequence \eqref{blockedRun}, but not the sequence \eqref{unblockedRun}, first consider the case $\ell > 1$.
Then \eqref{unblockedRun} is not affected by the blocking clause, as it requires less than $s+\ell$ steps.
If $\ell = 1$, then the loop that needs to be blocked is a single \emph{original transition} (i.e., a transition that results from applying $\mbip$ to $\tau$) due to \Cref{remark:loops}.
So $\ind[\id]{x}{s} > 1$ is falsified by \eqref{blockedRun}, as $\tau_s$ is an original transition, i.e., using it for the $s^{th}$ step implies $\ind[\id]{x}{s} \doteq 1$.
However, $\ind[\id]{x}{s} > 1$ is satisfied by \eqref{unblockedRun}, as $\pi$ is a learned
transition, so using it implies
$\ind[\id]{x}{s} > 1$.

\begin{remark}[Extra Variables and Negation]
In \Cref{def:blocking}, $\mbip$ is used to project $\pi$ according to the model $\overline{\sigma}$.
In this way, negation has the intended effect, i.e., 
\[
  [\vec{x}/\vec{v},\vec{x}'/\vec{v}'] \models \neg\mbip(\ldots) \qquad \text{iff} \qquad \vec{v} \not\to_{\mbip(\ldots)} \vec{v}',
\]
as $\mbip(\ldots)$ has no extra variables.
To see why this is important here, consider the relation $\tau \Def n > 0 \land x' \doteq x + n$, where $n$ is an extra variable.
Then $0 \to_\tau 1$, but
\[
  \neg\tau[x/0,x'/1] = (n \leq 0 \lor x' \not\doteq x + n)[x/0,x'/1] = n \leq 0 \lor 1 \not\doteq n
\]
is satisfiable, so $\neg\tau$ is not a suitable characterization of $\not\to_\tau$.
The reason is that $n$ is implicitly existentially quantified in $\tau$.
So to characterize $\not\to_\tau$, we have to negate $\exists n.\ \tau$ instead of $\tau$, resulting in $\forall n.\ n \leq 0 \lor x' \not\doteq x + n$.
Then, as desired,
\[
  (\forall n.\ n \leq 0 \lor x' \not\doteq x + n)[x/0,x'/1] = \forall n.\ n \leq 0 \lor 1 \not\doteq n
\]
is invalid.
To avoid quantifiers, we eliminate extra variables via $\mbip$ instead.
\end{remark}

In \Cref{alg:block2}, a pair consisting of $s + \ell - 1$ and the blocking clause is added to $\blocked$.
The first component means that the blocking clause has to be added to the SMT encoding when the transition relation is unrolled for the $(s+\ell-1)^{th}$ time, i.e., when $b = s + \ell - 1$.
So blocking clauses are added to the SMT encoding ``on demand'' (in \Cref{alg:block1}) to block
loops that have been found on the trace at some point.
Afterwards, TRL backtracks to the last step before the loop in \Cref{alg:backtrack}.

\begin{remark}[Adding Blocking Clauses]
To see why blocking clauses must only be added to the SMT encoding
in the $(s+\ell-1)^{th}$ unrolling, assume $\pi \equiv \top$.
Then, e.g., $\blockingclause(1,2,\pi,\overline{\sigma}) \equiv \bot$.
This means that unrolling the transition relation twice is superfluous, as every state can be reached in a single step with $\pi$, so the diameter is $1$.
But after learning $\pi$ when $b=2$ and backtracking to $b=0$, adding such a blocking
clause too early (e.g., before the first unrolling of the transition relation) would
\emph{immediately} result in an unsatisfiable SMT problem.
\end{remark}

\begin{example}[Blocking Redundant Loops]
  \label{ex:redundant}
  Consider the model
  \[
    \sigma \Def [\ind{w}{1}/\ind{x}{1}/\ind{y}{1}/0,\ind[\id]{x}{1}/2, \quad
      \ind{w}{2}/0,\ind{x}{2}/\ind{y}{2}/2,\ind[\id]{x}{2}/1, \quad
      \ind{x}{3}/\ind{y}{3}/3]
  \]
  and assume that TRL has already learned the relation $\tau^+_\inc$ (i.e., $\vec{\pi} = [\tau,\tau^+_\inc]$).
  Moreover, assume that the trace is $[\tau^+_\inc,\tau_\inc]$, so that TRL detects the loop $\tau_\inc$.
  To check for non-redundancy, we instantiate the pre- and post-variables in $\tau^+_\inc$ according to $\sigma$, taking the renaming $\mu_{2}$ into account (note that here $s = 2$, $\ell = 1$, and $\mu_{s,\ell} = \mu_{2,1} = \mu_2$):
  \[
    \sigma(\mu_{2}(\tau^+_\inc)) = \tau^+_\inc[w/0,x/y/2,x'/y'/3] \equiv \top.
  \]
  So our sufficient criterion for non-redundancy fails, as $\tau_\inc$ is indeed redundant w.r.t.\ $\tau^+_\inc$.
  Thus, TRL records that the following blocking clause has to be added for the second
  unrolling (i.e., when $b = s + \ell -1 = 2$).
  \paper{\begin{align*}
            & \mu_{s,\ell}(\neg\mbip(\tau^+_\inc,\overline{\sigma})) \lor \ind[\id]{x}{s}
      > 1 \hspace{.4em} = \hspace{.4em} \mu_{s,\ell}(\neg\tau^+_\inc) \lor \ind[\id]{x}{s}
      > 1 \tag{as $\tau^+_\inc$ is a transition} \\
   {} = {} & \mu_{2}(\neg (w \doteq 0 \land x' > x \land x' - x \doteq y' - y)) \lor \ind[\id]{x}{2} > 1\\ 
    {} \equiv {} & (w \not\doteq 0 \lor x' \leq x \lor x' - x \not\doteq y' - y)[w/\ind{w}{2}, x/\ind{x}{2},y/\ind{y}{2}, x'/\ind{x}{3},y'/\ind{y}{3}] \lor \ind[\id]{x}{2} > 1 \\
    {} = {} & \ind{w}{2} \not\doteq 0 \lor \ind{x}{3} \leq \ind{x}{2} \lor \ind{x}{3} - \ind{x}{2} \not\doteq \ind{y}{3} - \ind{y}{2} \lor \ind[\id]{x}{2} >1
  \end{align*}}
  \report{\begin{align*}
            & \mu_{s,\ell}(\neg\mbip(\tau^+_\inc,\overline{\sigma})) \lor \ind[\id]{x}{s} > 1 \hspace{.4em} = \hspace{.4em} \mu_{s,\ell}(\neg\tau^+_\inc) \lor \ind[\id]{x}{s} > 1 \tag{as $\tau^+_\inc$ is a transition} \\
    {} = {} & \mu_{2}(\neg (w \doteq 0 \land x' > x \land x' - x \doteq y' - y)) \lor
    \ind[\id]{x}{2} > 1\\
    {} \equiv {} & (w \not\doteq 0 \lor x' \leq x \lor x' - x \not\doteq y' - y)[w/\ind{w}{2}, x/\ind{x}{2},y/\ind{y}{2}, x'/\ind{x}{3},y'/\ind{y}{3}] \lor \ind[\id]{x}{2} > 1 \\
    {} = {} & \ind{w}{2} \not\doteq 0 \lor \ind{x}{3} \leq \ind{x}{2} \lor \ind{x}{3} - \ind{x}{2} \not\doteq \ind{y}{3} - \ind{y}{2} \lor \ind[\id]{x}{2} >1
  \end{align*}}  
  As this blocking clause is falsified by $\sigma$, it prevents TRL from finding the same model again after backtracking in \Cref{alg:backtrack}, so that TRL makes progress.
\end{example}
The following theorem states that our approach is sound.
\begin{restatable}
  {theorem}{soundness}
  \label{thm:soundness}
  If $\text{TRL}(\TT)$ returns $\safe$, then $\TT$ is safe.
\end{restatable}
\paper{
  \begin{proofsketch}
    The key idea of the proof, which works by induction over the length of runs, is to show that blocking clauses do not prevent us from reaching all reachable states.
    See \cite{arxiv} for the full proof.
  \end{proofsketch}
}
\makeproof*{thm:soundness}{
  \soundness*
  \begin{proof}
    Consider the SMT problem that is checked in \Cref{alg:safe}.
    In the $m^{th}$ iteration (starting with $m=1$), this problem is of the form $\varphi(m) \Def$
    \begin{align*}
       & \overbrace{\mu_{1}(\psi_\init)}^{\substack{\text{initial states}                                             \\
      \text{\Cref{alg:init}}}} \land                                                                                                 \\
                      & \overbrace{\bigwedge_{j=1}^{b(m)}}^{\substack{\text{one conjunct}                                          \\
        \text{per step}}}
      \left( \overbrace{(\ind[\id]{x}{j} \doteq 1 \lor \ind[\id]{x}{j} \not\doteq \ind[\id]{x}{j-1})}^{\substack{\text{transitivity} \\
          \text{\Cref{alg:trans}}}}
      \land \overbrace{\bigvee_{i=1}^{\ell'(m)} \mu_{j}(\pi_i \land x_{\id} \doteq i)}^{\substack{\text{transition relation}       \\
          \text{\Cref{alg:unroll}}}} \land \overbrace{\bigwedge_{(j,\pi) \in \blocked(m)}
      \pi}^{\mathclap{\substack{\text{blocking clauses}                                                                              \\
            \text{\Cref{alg:block1}}}}} \right)
    \end{align*}
    where $b(m)$, $\ell'(m)$, and $\blocked(m)$ are the value of $b$, the length of $\vec{\pi}$, and the values of $\blocked$ in \Cref{alg:safe} in the $m^{th}$ iteration, respectively.
    For simplicity, here we use an additional variable $\ind[\id]{x}{0}$ which only occurs in the first transitivity constraint (which is not generated by \Cref{alg}):
    \[
      \ind[\id]{x}{1} \doteq 1 \lor \ind[\id]{x}{1} \not\doteq \ind[\id]{x}{0}
    \]
    Therefore, this constraint is trivially satisfiable (e.g., by setting $\ind[\id]{x}{0}$ to $1$).

    Assume that $\TT$ is unsafe, but \Cref{alg} returns $\safe$.
    Then there is some $k \in \NN$ such that $\varphi(k)$ is unsatisfiable, and $\varphi(k')$ is satisfiable for all $1 \leq k' < k$ (otherwise, \Cref{alg} would have returned $\safe$ in an earlier iteration).

    As \Cref{alg} backtracks in \Cref{alg:backtrack}, we consider the sequence of natural numbers $1 \leq i_1 < \ldots < i_{b(k)} = k$ such that for all $1 \leq c \leq b(k)$, $i_c$ is the last iteration where $b=c$ in \Cref{alg:safe}.
    Then for all $1 \leq B \leq b(k)$, we have $\varphi(i_B) = \phi(B)$ where
    \begin{align*}
      \phi(B) \Def & \mu_{1}(\psi_\init) \land \\
                   & \bigwedge_{j=1}^{B}
      \left( (\ind[\id]{x}{j} \doteq 1 \lor \ind[\id]{x}{j} \not\doteq \ind[\id]{x}{j-1}) \land \bigvee_{i=1}^{\ell(B)}\mu_{j}(\pi_i \land x_{\id} \doteq i) \land \bigwedge_{(j,\pi) \in \blocked(k)} \pi \right).
    \end{align*}
    Here, we have $\ell(B) \Def \ell'(i_B)$.
    A blocking clause $\pi$ is only added to the SMT encoding if $(j,\pi) \in \blocked$ and $b=j$ (see \Cref{alg:block1}) and \Cref{alg} backtracks until $b \leq j$ whenever such an element is added to $\blocked$ (see \Cref{alg:backtrack}).
    So when only considering the iterations $i_1, \ldots, i_{b(k)}$, then $b(i_B) = B$ and all blocking clauses of the form $(j,\pi) \in \blocked(k)$ where $j < B$ are already present when unrolling the transition relation for the $B^{th}$ time, i.e., they are already contained in $\blocked(i_B)$.

    In other words, we have
    \[
      \{(j,\pi) \mid (j,\pi) \in \blocked(k) \mid j < B\} \subseteq \blocked(i_B).
    \]
    Thus, we may use $\blocked(k)$ instead of $\blocked(i_B)$ in the definition of
    $\phi$.

    Let $c \in \NN$ and $\vec{v}_0,\ldots\vec{v}_c \in \CC^d$ be arbitrary but fixed where $[\vec{x}/\vec{v}_0] \models \psi_\init$ and
    \[
      \vec{v}_0 \to_{\tau} \ldots \to_{\tau} \vec{v}_c.
    \]
    We use induction on $c$ to show that\footnote{While $\phi(B)$ only corresponds to a formula that is checked by \Cref{alg} if $B > 0$, $\phi(0) \equiv \mu_1(\psi_\init)$ is well defined, too.}
    \begin{multline}
      \label{eq:goal}
      \forall 0 \leq i \leq c.\ \exists B(i) < b(k), h(0,i) < \ldots < h(B(i),i).\ h(0,i) = 0 \land h(B(i),i) = i\\
      {} \land \phi(B(i)) \text{ is consistent with } [\mu_{j+1}(\vec{x})/\vec{v}_{h(j,i)} \mid 0 \leq j \leq B(i)].
    \end{multline}
    Intuitively, $B(i)$ is the number of steps that are needed to reach $\vec{v}_i$ when also using learned relations, and $\vec{v}_{h(j,i)}$ is the $j^{th}$ state in the resulting run that leads to $\vec{v}_i$.
    Thus, \eqref{eq:goal} shows that for all (arbitrary long) runs that start in a state satisfying $\psi_\init$, all states that are reachable with $\to_\tau$ in arbitrarily many steps can also be reached in less than $b(k)$ steps (i.e., in constantly many steps) if one may also use the (transitive) learned relations.
    Of course, this only holds provided that \Cref{alg} returns $\safe$ in the $k^{th}$ iteration.

    Once we have shown \eqref{eq:goal}, we can prove the theorem:
    We had assumed that $\TT$ is unsafe, i.e., that there is a reachable error state $\vec{v}_c$ and that $\varphi(k) = \varphi(i_{b(k)}) = \phi(b(k))$ is unsatisfiable.
    We use \eqref{eq:goal} for $i = c$:
    By \eqref{eq:goal}, there is some $B(c) < b(k)$ such that
    \[
      \phi(B(c)) \text{ is consistent with } [\mu_{j+1}(\vec{x})/\vec{v}_{h(j,c)} \mid 0 \leq j \leq B(c)].
    \]
    Moreover, as $\vec{v}_c$ is an error state, $\psi_\err$ is consistent with $[\vec{x}/\vec{v}_{c}]$ and hence,
    \[
      \mu_{B(c)+1}(\psi_\err) \text{ is consistent with } [\mu_{B(c)+1}(\vec{x})/\vec{v}_{c}] = [\mu_{B(c)+1}(\vec{x})/\vec{v}_{h(B(c),c)}].
    \]
    Thus,
    \[
      \phi(B(c)) \land \mu_{B(c)+1}(\psi_\err) \text{ is consistent with } [\mu_{j+1}(\vec{x})/\vec{v}_{h(j,c)} \mid 0 \leq j \leq B(c)].
    \]
    Hence,
    \begin{equation}
      \label{eq:contradiction}
      \text{\Cref{alg} returns $\unknown$ in \Cref{alg:err2} in iteration $i_{B(c)+1}$}
    \end{equation}
    or earlier.
    The reason is that we have $b = B(c)+1$ in iteration $i_{B(c)+1}$, so in this iteration \Cref{alg} checks satisfiability of the formula $\phi(B(c)) \land \mu_{B(c)+1}(\psi_\err)$ in \Cref{alg:err2}.
    As we have $b(k) > B(c)$, we get $k = i_{b(k)} \geq i_{B(c)+1}$.
    Hence, \eqref{eq:contradiction} contradicts the assumption that \Cref{alg} returns $\safe$ in \Cref{alg:safe} in iteration $k$. 

    We now prove \eqref{eq:goal}. 
    In the induction base, we have
    \begin{align*}
                             & [\vec{x}/\vec{v}_0] \models \psi_\init                                                     \\
      {} \curvearrowright {} & [\mu_{1}(\vec{x})/\vec{v}_0] \models \mu_{1}(\psi_\init)                                   \\
      {} \curvearrowright {} & [\mu_{1}(\vec{x})/\vec{v}_0] \models \phi(0) \tag{as $\phi(0) \equiv \mu_{1}(\psi_\init)$}
    \end{align*}
    Hence, the claim follows for
    \[
      B(0) = 0 = h(0,0).
    \]
    In the induction step, the induction hypothesis implies
    \begin{multline}
      \label{eq:IH}
      \forall 0 \leq i < c.\ \exists B(i) < b(k), h(0,i) < \ldots < h(B(i),i).\ h(0,i) = 0 \land h(B(i),i) = i\\
      {} \land \phi(B(i)) \text{ is consistent with } [\mu_{j+1}(\vec{x})/\vec{v}_{h(j,i)} \mid 0 \leq j \leq B(i)].
    \end{multline}
    Let:
    \begin{align*}
      I & {} \Def \min \{i \mid 0 \leq i < c, 1 \leq j \leq \ell(B(i)+1), \vec{v}_i \to_{\pi_j} \vec{v}_c\} \\
      J & {} \Def \max \{j \mid 1 \leq j \leq \ell(B(I)+1), \vec{v}_I \to_{\pi_j} \vec{v}_c\}
    \end{align*}
    So $I$ is the minimal index such that we can make a step from $\vec{v}_I$ to $\vec{v}_c$, and among the relational formulas that can be used for this step, $\pi_J$ is the one that was learned last.
    Note that $I$ (and hence also $J$) exists, as we have $\vec{v}_{c-1} \to_{\tau} \vec{v}_c$ and $\tau = \pi_1$.
    By \eqref{eq:IH},
    \[
      \phi(B(I)) \text{ is consistent with } [\mu_{j+1}(\vec{x})/\vec{v}_{h(j,I)} \mid 0 \leq j \leq B(I)].
    \]
    Let $\theta_I$ be an extension of $[\mu_{j+1}(\vec{x})/\vec{v}_{h(j,I)} \mid 0 \leq j \leq B(I)]$ such that $\theta_I \models \phi(B(I))$, and let $\theta$ be an extension of
    \begin{equation}
      \label{eq:theta}
      \theta_I \uplus [\mu_{B(I)+2}(\vec{x})/\vec{v}_c] \uplus [\mu_{B(I)+1}(x_{\id}) / J]
    \end{equation}
    such that $\theta \models \mu_{B(I)+1}(\pi_J)$, which exists as we have:
    \begin{align*}
      & \theta(\mu_{B(I)+1}(\vec{x})) \\
      {} = {} & \theta_I(\mu_{B(I)+1}(\vec{x})) \tag{by \eqref{eq:theta}} \\
      {} = {} & \vec{v}_{h(B(I),I)} \tag{def.\ of $\theta_I$} \\
      {} = {} & \vec{v}_I  \tag{as $h(B(I),I) = I$} \\
      {} \to_{\pi_J} {} & \vec{v}_c \tag{def.\ of $J$}\\
      {} = {} & \mu_{B(I)+1}(\vec{x}')[\mu_{B(I)+1}(\vec{x}')/\vec{v}_c] \\
      {} = {} & \mu_{B(I)+1}(\vec{x}')[\mu_{B(I)+2}(\vec{x})/\vec{v}_c] \tag{as $\mu_{B(I)+2}(\vec{x}) = \mu_{B(I)+1}(\vec{x}')$} \\
      {} = {} & \theta(\mu_{B(I)+1}(\vec{x}')) \tag{by \eqref{eq:theta}}
    \end{align*}
    So we have
    \[
      \theta(\ind{\vec{x}}{B(I)+1}) = \theta_I(\ind{\vec{x}}{B(I)+1}) = \vec{v}_{h(B(I),I)} = \vec{v}_I
    \]
    and
    \[
      \theta(\ind{\vec{x}}{B(I)+2}) = \vec{v}_c.
    \]

    We now show $\theta \models \phi(B(I)+1)$.
    Then we get $B(c) = B(I) + 1$, $h(j,c) = h(j,I)$ for all $j \leq B(I)$, and $h(B(c),c) = c$, which finishes the proof of \eqref{eq:goal}.

    Since $\theta$ is an extension of $\theta_I$, we have $\theta \models \phi(B(I))$, so we only need to show
    \[
   \hspace*{-.8cm}   \theta \models (\ind[\id]{x}{B(I)+1} \doteq 1 \lor \ind[\id]{x}{B(I)+1} \not\doteq
      \ind[\id]{x}{B(I)}) \land \hspace*{-.2cm}\bigvee_{i=1}^{\ell(B(I)+1)}\hspace*{-.2cm}\mu_{B(I)+1}(\pi_i \land
      x_{\id} \doteq i) \land \bigwedge_{\mathclap{(B(I)+1,\pi) \in \blocked(k)}} \pi, 
    \]
    i.e., we only need to show that $\theta$ is a model of the last conjunct of $\phi(B(I)+1)$.
    For the disjunction
    \begin{equation}
      \label{eq:disjunction}
      \bigvee_{i=1}^{\ell(B(I)+1)}\mu_{B(I)+1}(\pi_i \land x_{\id} \doteq i),
    \end{equation}
    note that $J \leq \ell(B(I)+1)$.
    Thus, to show $\theta \models \eqref{eq:disjunction}$, it suffices if
    \[
      \theta \models \mu_{B(I)+1}(\pi_J \land x_{\id} \doteq J),
    \]
    which holds by construction of $\theta$.
    Thus, it remains to show
    \[
      \theta \models (\ind[\id]{x}{B(I)+1} \doteq 1 \lor \ind[\id]{x}{B(I)+1} \not\doteq \ind[\id]{x}{B(I)}) \land \bigwedge_{\mathclap{(B(I)+1,\pi) \in \blocked(k)}} \pi.
    \]

    We first consider the disjunction $\ind[\id]{x}{B(I)+1} \doteq 1 \lor \ind[\id]{x}{B(I)+1} \not\doteq \ind[\id]{x}{B(I)}$.
    We show that we always have $J=1$ or $\theta(\ind[\id]{x}{B(I)}) \neq J$.
    Then this disjunction is clearly satisfied by $\theta$ since $\theta(\ind[\id]{x}{B(I)+1}) = J$.
    To see why $J=1$ or $\theta(\ind[\id]{x}{B(I)}) \neq J$ holds, assume that $\theta(\ind[\id]{x}{B(I)}) = J > 1$.
    Then:
    \begin{align*}
                             & \theta \models \phi(B(I)) \tag {as $\theta$ is an extension of $\theta_I$} \\
      {} \curvearrowright {} & \theta \models \mu_{B(I)}(\pi_J \land x_{\id} \doteq J) \tag{as $\theta(\ind[\id]{x}{B(I)}) = J$ by assumption}                                                                                 \\
      {} \curvearrowright {} & \mu_{B(I)}(\pi_J) \text{ is consistent with } [\mu_{j+1}(\vec{x})/\vec{v}_{h(j,I)} \mid 0 \leq j \leq B(I)] \tag{as $[\mu_{j+1}(\vec{x})/\vec{v}_{h(j,I)} \mid 0 \leq j \leq B(I)] \subseteq \theta$} \\
      {} \curvearrowright {} & \mu_{B(I)}(\pi_J) \text{ is consistent with } [\mu_{B(I)}(\vec{x})/\vec{v}_{h(B(I)-1,I)},\mu_{B(I)+1}(\vec{x})/\vec{v}_{h(B(I),I)}] \tag{by instantiating $j$ with $B(I)-1$ and $B(I)$}           \\
      {} \curvearrowright {} & \pi_J \text{ is consistent with } [\vec{x}/\vec{v}_{h(B(I)-1,I)},\vec{x}'/\vec{v}_{h(B(I),I)}]                                                                                                      \\
      {} \curvearrowright {} & \pi_J \text{ is consistent with } [\vec{x}/\vec{v}_{h(B(I)-1,I)},\vec{x}'/\vec{v}_{I}] \tag{as $h(B(I),I) = I$}                                                                                     \\
      {} \curvearrowright {} & \vec{v}_{h(B(I)-1,I)} \to_{\pi_J} \vec{v}_{I}.
    \end{align*}
    By the definition of $I$ and $J$, we also have $\vec{v}_{I} \to_{\pi_J} \vec{v}_{c}$.
    So by transitivity of $\to_{\pi_J}$ (which holds since $J > 1$), we get $\vec{v}_{h(B(I)-1,I)} \to_{\pi_J} \vec{v}_{c}$.
    As we have $h(B(I)-1,I) < I$ by definition of $h$, this contradicts minimality of $I$.
    Note that here is the only point where we need the transitivity of learned relations.

    Therefore, it remains to show
    \[
      \theta \models \bigwedge_{\mathclap{(B(I)+1,\pi) \in \blocked(k)}} \pi.
    \]
    The elements of $\blocked(k)$ have the form $(s + \ell - 1, \blockingclause(s,\ell,\pi_j,\overline{\sigma}))$ where $[\tau_i]_{i=s}^{s+\ell-1}$ is a loop on the trace, $\pi_j \in \tail(\vec{\pi})$ is a learned relation, and $\overline{\sigma} \models \pi_j$.
    Moreover, we have $s + \ell - 1 = B(I) + 1$.
    We now perform a case analysis for the two cases where $\ell = 1$ and $\ell > 1$.

    In Case 1 (where $\ell = 1$), the blocking clause has the form
    $\blockingclause(s,1,\linebreak \pi_j,\overline{\sigma})$ for $s = B(I)+1$.
    Thus, we show that $\theta$ cannot violate a blocking clause of the form
    \[
      \blockingclause(B(I),1,\pi_j,\overline{\sigma}) = \mu_{B(I)+1}(\neg\mbip(\pi_j,\overline{\sigma})) \lor \ind[\id]{x}{B(I)+1} > 1.
    \]
    To see this, we first consider the case where the negation of the first disjunct holds and prove that then the second disjunct is true.
    The reason is that we have:
    \begin{align*}
                             & \theta \models \mu_{B(I)+1}(\mbip(\pi_j, \overline{\sigma}))                                                                                     \\
      {} \curvearrowright {} & \theta \circ \mu_{B(I)+1} \text{ is consistent with } \pi_j \tag{since $\mbip(\pi_j,\overline{\sigma}) \models \pi_j$ by \Cref{def:projections}} \\
      {} \curvearrowright {} & \theta(\mu_{B(I)+1}(\vec{x})) \to_{\pi_j} \theta(\mu_{B(I)+1}(\vec{x}'))                                                                \\
      {} \curvearrowright {} & \theta(\mu_{B(I)+1}(\vec{x})) \to_{\pi_j} \theta(\mu_{B(I)+2}(\vec{x}))                                                                \\
      {} \curvearrowright {} & \vec{v}_{h(B(I),I)} \to_{\pi_j} \vec{v}_{c} \tag{def.\ of $\theta$}                                                                 \\
      {} \curvearrowright {} & \vec{v}_{I} \to_{\pi_j} \vec{v}_{c} \tag{as $h(B(I),I) = I$}
    \end{align*}
    So the step from $\vec{v}_{I}$ to $\vec{v}_{c}$ can be done by a learned relation $\pi_j$.
    As $\pi_j$ is a learned relation, we have $j > 1$.
    This implies $\theta(\ind[\id]{x}{B(I)+1}) = J > 1$, as $J$ is maximal and hence $J \geq j > 1$.

    Now we consider the case where the negation of the first disjunct does not hold and prove that then the first disjunct is true.
    The reason is that due to the completeness of $\AA$, $\theta \centernot\models \mu_{B(I)+1}(\mbip(\pi_j, \overline{\sigma}))$ implies $\theta \models \mu_{B(I)+1}(\neg\mbip(\pi_j, \overline{\sigma}))$.
    So in this case the blocking clause is satisfied as well.

    In Case 2 (where $\ell > 1$), we show that $\theta$ cannot violate a blocking clause of the form
    \[
      \blockingclause(s,\ell,\pi_j,\overline{\sigma}) = \mu_{s,\ell}(\neg\mbip(\pi_j, \overline{\sigma}))
    \]
    where $s + \ell - 1 = B(I)+1$, i.e., $s + \ell = B(I) + 2$.
    The reason is that we have
    \begin{align*}
                             & \theta \centernot\models \mu_{s,\ell}(\neg\mbip(\pi_j, \overline{\sigma}))                                                                       \\
      {} \curvearrowright {} & \theta \models \mu_{s,\ell}(\mbip(\pi_j, \overline{\sigma})) \tag{as $\AA$ is complete} \\
      {} \curvearrowright {} & \theta \circ \mu_{s,\ell} \text{ is consistent with } \pi_j \tag{since $\mbip(\pi_j,\overline{\sigma}) \models \pi_j$ by \Cref{def:projections}} \\
      {} \curvearrowright {} & \theta(\mu_{s,\ell}(\vec{x})) \to_{\pi_j} \theta(\mu_{s,\ell}(\vec{x}'))                                                              \\
      {} \curvearrowright {} & \vec{v}_{h(s-1,I)} \to_{\pi_j} \vec{v}_{c} \tag{def.\ of $\theta$, as $s+\ell = B(I)+2$}
    \end{align*}
    For the last step, note that
    \begin{align*}
      \theta(\mu_{s,\ell}(\vec{x})) & {} = \theta(\mu_{s}(\vec{x})) = \theta_I(\mu_{s}(\vec{x})) = \vec{v}_{h(s-1,I)} & \text{and} \\
      \theta(\mu_{s,\ell}(\vec{x}')) & {} = \theta(\ind{\vec{x}}{s + \ell}) = \theta(\ind{\vec{x}}{B(I)+2}) = \vec{v}_{c}.
    \end{align*}
    We have $s - 1 < B(I)$ and thus $h(s-1,I) < h(B(I),I) = I$, which contradicts minimality of $I$.
    This finishes the proof of \eqref{eq:goal}.
    \qed
  \end{proof}
}

\begin{remark}[Termination]
  \label{remark:termination}
In general, \Cref{alg} does not terminate, since
 $\tip$ decomposes the relation into finitely many cases and
approximates their transitive closures independently, but
${\to^+_{\tau_\Loop}} \subseteq \bigcup_{\sigma \models \tau_{\Loop}}
{\to_{\tip(\tau_{\Loop},\sigma)}}$ is not guaranteed (\Cref{remark:properties-tip}).
To see why this may prevent termination, consider a loop $\tau_{\Loop}$ and assume that there are reachable states
$\vec{v},\vec{v}'$ with $\vec{v} \to_{\tau_\Loop}^+ \vec{v}'$, but $\vec{v}
\not\to_{\tip(\tau_{\Loop},\sigma)} \vec{v}'$ for all models $\sigma$ of
$\tau_\Loop$.
Then TRL may find a model that corresponds to a run from $\vec{v}$ to $\vec{v}'$.
Unless $\vec{v}$ can be evaluated to $\vec{v}'$ with another learned transition $\pi \notin \{\tip(\tau_{\Loop},\sigma) \mid \sigma \models \tau_{\Loop}\}$ by coincidence, this loop cannot be blocked and TRL learns a new relation.
Thus, TRL may keep learning new relations as long as there are loops whose transitive closure is not yet covered by learned relations.

As the elements of $\{\tip(\tau,\sigma) \mid \sigma \models \tau\}$ are independent of each other, a more ``global'' view may help to enforce convergence.
We leave that to future work.
\end{remark}

\paper{
  \begin{example}[\Cref{ex:ex1} Finished]
    After learning $\tau^+_\dec$ and $\tau^+_\inc$, the underlying SMT problem becomes unsatisfiable when $b=3$ after adding appropriate blocking clauses, so that $\tau^+_\dec$ and $\tau^+_\inc$ are preferred over $\tau_\dec$ and $\tau_\inc$.
    The reason is that $\tau^+_\dec$ and $\tau^+_\inc$ must not be used twice in a row due to \Cref{alg:trans} of \Cref{alg}, and $\tau^+_\inc$ cannot be used after $\tau^+_\dec$, as it requires $w \doteq 0$, but $\tau^+_\dec$ sets $w$ to $1$.
    Thus, \Cref{alg} returns $\safe$.
    See \cite{arxiv} for a detailed run of
\Cref{alg} 
on \Cref{ex:ex1}.
  \end{example}
}

\report{
\subsection{A Complete Example}
\label{sec:example}

For a complete run of TRL on our example, assume that we obtain the following traces (where the detected loops are underlined):
\begin{enumerate}
  \item \label{it:a}
        $[\underline{\tau_\inc}]$, resulting in the learned relation $\tau^+_\inc$ and the blocking clause $\mu_{1}(\neg\tau^+_\inc \lor x_\id > 1)$ which ensures that if $b = 1$, then we cannot use $\tau_\inc$ but would have to use $\tau^+_\inc$.
  \item \label{it:b}
        $[\underline{\tau_\dec}]$, resulting in the learned relation $\tau^+_\dec$ and the blocking clause $\mu_{1}(\neg\tau^+_\dec \lor x_\id > 1)$ which ensures that if $b = 1$, then we cannot use $\tau_\dec$ but would have to use $\tau^+_\dec$.
  \item \label{it:c}
        $[\tau^+_\inc,\underline{\tau_\inc}]$, resulting in the blocking clause $\mu_{2}(\neg\tau^+_\inc \lor x_\id > 1)$ which ensures that if $b = 2$, then we cannot use $\tau_\inc$.
        Using $\tau^+_\inc$ twice after each other is also not possible due to transitivity (\Cref{alg:trans}).
  \item \label{it:d}
        $[\tau^+_\dec,\underline{\tau_\dec}]$, resulting in the blocking clause $\mu_{2}(\neg\tau^+_\dec \lor x_\id > 1)$ which ensures that if $b = 2$, then we cannot use $\tau_\dec$.
        Using $\tau^+_\dec$ twice after each other is also not possible due to transitivity (\Cref{alg:trans}).
  \item \label{it:e}
        $[\tau^+_\inc, \tau^+_\dec, \underline{\tau_\dec}]$, resulting in the blocking clause $\mu_{3}(\neg\tau^+_\dec \lor x_\id > 1)$ which ensures that if $b = 3$, then we cannot use $\tau_\dec$.
\end{enumerate}
Now we are in the following situation:
\begin{itemize}
  \item The first element of the trace cannot be $\tau_\inc$ or $\tau_\dec$ due to \eqref{it:a} and \eqref{it:b}.
  \item If the first element of the trace is $\tau^+_\inc$, then: 
        \begin{itemize}
          \item The second element of the trace cannot be $\tau_\inc$ or $\tau_\dec$ due to \eqref{it:c} and \eqref{it:d}.
          \item The second element of the trace cannot be $\tau^+_\inc$ due to \Cref{alg:trans}.
          \item If the second element of the trace is $\tau^+_\dec$, then:
                \begin{itemize}
                  \item The third element of the trace cannot be $\tau_\inc$ or $\tau^+_\inc$, as $\tau^+_\dec$ sets $w$ to $1$, but $\tau_\inc$ and $\tau^+_\inc$ require $w = 0$.
                  \item The third element of the trace cannot be $\tau_\dec$ due to \eqref{it:e}.
                  \item The third element of the trace cannot be $\tau^+_\dec$ due to \Cref{alg:trans}.
                \end{itemize}
        \end{itemize}
        So this case becomes infeasible with the $3^{rd}$ unrolling of the transition relation.
  \item If the first element of the trace is $\tau^+_\dec$, then:
        \begin{itemize}
          \item The second element of the trace cannot be $\tau_\inc$ or $\tau^+_\inc$, as $\tau^+_\dec$ sets $w$ to $1$, but $\tau_\inc$ and $\tau^+_\inc$ require $w = 0$.
          \item The second element of the trace cannot be $\tau_\dec$ due to \eqref{it:d}.
          \item The second element of the trace cannot be $\tau^+_\dec$ due to \Cref{alg:trans}.
        \end{itemize}
        So this case becomes infeasible with the $2^{nd}$ unrolling of the transition relation.
\end{itemize}
Thus, the underlying SMT problem becomes unsatisfiable when $b=3$, such that $\safe$ is returned in \Cref{alg:safe}.
}

%% file: recurrence.tex
\section{Implementing $\tip$ for Linear Integer Arithmetic}
\label{sec:rec}

We now explain how to compute transitive projections for quantifier-free linear integer arithmetic via recurrence analysis.
As in SMT-LIB \cite{smtlib}, in our setting linear integer arithmetic also features (in)divisibility predicates of the form $e|t$ (or $e\!{\not|}t$) where $e \in \NN_+$ and $t$ is an integer-valued term.
Then we have $\sigma \models e|t$ iff $\sigma(t)$ is a multiple of $e$, and $\sigma \models e\!{\not|}t$, otherwise.

The technique that we use is inspired by the recurrence analysis from \cite{kincaid15}.
However, there are some important differences.
The approach from \cite{kincaid15} computes convex hulls to over-approximate disjunctions by conjunctions, and it relies on polyhedral projections.
In our setting, we always have a suitable model at hand, so that we can use $\mbip$ instead.
Hence, our recurrence analysis can be implemented more efficiently.\footnote{The \emph{double description method}, which is popular for computing polyhedral projections and convex hulls, and other state-of-the-art approaches have exponential complexity \cite{dd-exp,fmplex}.
  See \cite{convex-hull} for an easily accessible discussion of the complexity of the double description method.
  In contrast, combining the model based projection from \cite{spacer} with syntactic implicant projection \cite{adcl} yields a polynomial time algorithm for $\mbip$.}
Additionally, our recurrence analysis can handle divisibility predicates, which are not covered in \cite{kincaid15}.

On the other hand, \cite{kincaid15} yields an over-approximation of the transitive closure of the given relation, whereas our approach performs an implicit case analysis (via $\mbip$) and only yields an over-approximation of the transitive closure of one out of finitely many cases.

Moreover, the recurrence analysis from \cite{kincaid15} also discovers non-linear relations, and then uses linearization techniques to eliminate them.
For simplicity, our recurrence analysis only derives linear relations so far.
However, just like \cite{kincaid15}, we could also derive non-linear relations and linearize them afterwards.
Apart from these differences, our technique is analogous to \cite{kincaid15}.

In the sequel, let $\tau$ and $\sigma \models \tau$ be fixed.
Our implementation of $\tip(\tau,\sigma)$ first searches for \emph{recurrent literals}, i.e., literals of the form\footnote{W.l.o.g., we assume that literals are never negated, as we can negate the corresponding (in)equalities or divisibility predicates directly instead.
  Furthermore, in our implementation, we replace disequalities $s \not\doteq t$ with $s >
  t \lor s < t$ and eliminate the resulting disjunction via $\mbip$ to obtain a transition
  without disequalities.}
\[
  t \bowtie 0 \text{ or } e|t \quad \text{where} \quad t = \sum_{x \in \vec{x}} c_x
  \cdot (x'-x) + c, \; {\bowtie} \in \{\leq,\geq,<,>,\doteq\},
  \text{ and } c_x,c \in \ZZ.
\]
Hence, these literals provide information about the change of values of variables.
To find such literals, we introduce a fresh variable $x_\delta$ for each $x \in \vec{x}$, and we conjoin $x_\delta \doteq x' - x$ to $\tau$, i.e., we compute
\[
  \tau_{\land \delta} \Def \tau \land \bigwedge_{x \in \vec{x}} x_\delta \doteq x' - x.
\]
So the value of $x_\delta$ corresponds to the change of $x$ when applying $\tau$.
Next, we use $\mbip$ to eliminate all variables but $\{x_\delta \mid x \in \vec{x}\}$ from $\tau_{\land\delta}$, resulting in $\tau_\delta$.
More precisely, we have
\[
  \tau_\delta \Def \mbp(\tau_{\land \delta}, \quad \sigma \uplus [x_\delta/\sigma(x'-x) \mid x \in \vec{x}], \quad \{x_\delta \mid x \in \vec{x}\}).
\]
Finally, to obtain a formula where all literals are recurrent, we replace each $x_\delta$ by its definition, i.e., we compute
\[
  \tau_\rec \Def \tau_\delta[x_\delta / x' - x \mid x \in \vec{x}].
\]
\begin{example}[Finding Recurrent Literals]
  \label{ex:finding-rec}
  Consider the transition $\tau_\dec$.
  We first construct the formula
  \[
    \tau_{\land \delta} \Def \tau_\dec \land w_\delta \doteq w' - w \land x_\delta \doteq x' - x \land y_\delta \doteq y' - y.
  \]
  Then for any model $\sigma \models \tau_\dec$, we get\footnote{In the case of $\tau_\dec$, we obtain the same formula $\tau_\delta$ for \emph{every} model $\sigma \models \tau_\dec$, as variables can simply be eliminated by propagating equalities.}
  \begin{align*}
    \tau_\delta \Def {} & \mbp(\tau_{\land \delta}, \quad \sigma \uplus [w_\delta/0, \; x_\delta/{-}1, \; y_\delta/{-}1], \quad \{w_\delta, x_\delta, y_\delta\}) \\
                {} = {} & w_\delta \doteq 0 \land x_\delta \doteq -1 \land y_\delta \doteq -1.
  \end{align*}
  Next, replacing $w_\delta,x_\delta$, and $y_\delta$ with their definition results in
  \begin{align*}
    \tau_\rec \Def {} & w' - w \doteq 0 \land x' - x \doteq -1 \land y' - y \doteq -1        \\
    \equiv {}         & w' - w \doteq 0 \land x' - x + 1 \doteq 0 \land y' - y + 1 \doteq 0.
  \end{align*}
\end{example}
Then the construction of $\tip(\tau,\sigma)$ proceeds as follows:
\begin{itemize}
  \item $\tip(\tau,\sigma)$ contains the literal $n > 0$, where $n \in \VV$ is a fresh extra variable
  \item for each literal $\sum_{x \in \vec{x}} c_x \cdot (x'-x) + c \bowtie 0$ of $\tau_\rec$, $\tip(\tau,\sigma)$ contains the literal $\sum_{x \in \vec{x}} c_x \cdot (x'-x) + n \cdot c \bowtie 0$
  \item for each literal $e|\sum_{x \in \vec{x}} c_x \cdot (x'-x) + c$ of $\tau_\rec$, $\tip(\tau,\sigma)$ contains the literal $e|\sum_{x \in \vec{x}} c_x \cdot (x'-x) + n \cdot c$
\end{itemize}
Intuitively, the extra variable $n$ can be thought of as a ``loop counter'', i.e., when
$n$ is instantiated with some constant $k$, then the literals above approximate the change
of variables when $\to_\tau$ is applied $k$ times.

\begin{example}[Computing $\tip$ (1)]
  \label{ex:tip1}
  Continuing \Cref{ex:finding-rec}, $\tip(\tau_\dec,\sigma)$ contains the literals
  \begin{align*}
    {}           & n > 0 \land w' - w + n \cdot 0 \doteq 0 \land x' - x + n \cdot 1 \doteq 0 \land y' - y + n \cdot 1 \doteq 0 \\
    {} \equiv {} & n > 0 \land w' \doteq w \land x' \doteq x - n \land y' \doteq y - n
  \end{align*}
  for any model $\sigma \models \tau_\dec$.
  Note that in our example, this formula precisely characterizes the change of the variables after $n$ iterations of $\to_{\tau_\dec}$.
  To simplify the formula above, we can propagate the equality $n = x - x'$, resulting in:
  \begin{align}
    & x - x' > 0 \land w' \doteq w \land y' \doteq y - x + x' \notag \\
    {} \equiv {} & w' \doteq w \land x' < x \land x' - x \doteq y' - y \label{eq:rel}
  \end{align}
\end{example}
Compared to $\tau_\dec^+$, \eqref{eq:rel} lacks the literal $w \doteq 1$.
To incorporate information about the pre- and post-variables (but not about their relation) we conjoin $\mbp(\tau,\sigma,\vec{x})$ and $\mbp(\tau,\sigma,\vec{x}')$ to $\tip(\tau,\sigma)$.

\begin{example}[Computing $\tip$ (2)]
  We finish \Cref{ex:tip1} by conjoining
  \[
    \mbp(\tau_\dec,\sigma,\{w,x,y\}) = w \doteq 1 \qquad \text{and} \qquad \mbp(\tau_\dec,\sigma,\{w',x',y'\}) = w' \doteq 1
  \]
  to \eqref{eq:rel}, resulting in:
  \[
    w' \doteq w \land x' < x \land x' - x \doteq y' - y \land w \doteq 1 \land w' \doteq 1 \quad \equiv \quad \tau_\dec^+
  \]
\end{example}

\begin{example}[Divisibility]
  To see how $\tip$ can handle divisibility predicates, consider the transition
  \[
    \tau \Def 2|x \land 3|x' - x + 1.
  \]
  Then our approach identifies the recurrent literal $\tau_\rec = 3|x' - x + 1$, so that $\tip(\tau,\sigma)$ contains the literal $3|x' - x + n$.
  To see why we conjoin this literal to $\tip(\tau,\sigma)$, note that $3|x'-x+1$ and
  $3|x'-x+n$ are equivalent to $x'-x + 1 \equiv_3 0$ and $x'-x + n \equiv_3 0$, respectively, where ``$\equiv_3$'' denotes congruence modulo $3$.
  So $e \to^n_{\tau} e'$ implies $e'-e+n \equiv_3 0$, just like $e \to^n_{\pi} e'$ implies
  $e'-e+n=0$ for $\pi \Def x'-x+1 \doteq 0$.
  Moreover, we have $\mbp(\tau,\sigma,\vec{x}) = 2|x$, and thus
  \[
    \tip(\tau,\sigma) = n>0 \land 3|x'-x+n \land 2|x.
  \]
\end{example}

\noindent
We refer to \cite{arxiv} for
the straightforward proof of the following theorem:
\newcounter{tip}
\setcounter{tip}{\value{theorem}}
\begin{theorem}[Correctness of $\tip$]
  \label{thm:tip}
  The function $\tip$ as defined above is a transitive projection.
\end{theorem}
\makeproof*{thm:tip}{
  \setcounter{theorem}{\thetip}
  \begin{theorem}[Correctness of $\tip$]
    \label{thm:Correctness_tip}
    The function $\tip$ as defined in \Cref{sec:rec} is a transitive projection.
  \end{theorem}
  \begin{proof}
    We have to prove that
    \begin{itemize}
      \item[(a)] $\tip(\tau,\sigma)$ is consistent with $\sigma$,
      \item[(b)] $\{\tip(\tau,\theta) \mid \theta \models \tau\}$ is finite, and
      \item[(c)] $\to_{\tip(\tau,\sigma)}$ is transitive
    \end{itemize}
    for all transitions $\tau$ and all $\sigma \models \tau$.

    For item (a), we first prove $\sigma \models \tau_\rec$.
    We have:
    \begin{align*}
                             & \sigma \models \tau                                                                                                          \\
      {} \curvearrowright {} & \sigma \uplus [x_\delta / \sigma(x'-x) \mid x \in \vec{x}] \models \tau_{\land\delta} \tag{by def.\ of $\tau_{\land\delta}$} \\
      {} \curvearrowright {} & [x_\delta / \sigma(x'-x) \mid x \in \vec{x}] \models \tau_\delta \tag{by def.\ of $\mbp$ and $\tau_{\delta}$}                \\
      {} \curvearrowright {} & [x/\sigma(x),x'/\sigma(x') \mid x \in \vec{x}] \models \tau_\rec \tag{by def.\ of $\tau_{\rec}$}                             \\
      {} \curvearrowright {} & \sigma \models \tau_\rec
    \end{align*}
    Now we prove
    $\sigma' \Def \sigma \uplus [n/1] \models \tip(\tau,\sigma)$.
    To this end, we consider the literals that are added to $\tip(\tau,\sigma)$ by the procedure described in \Cref{sec:rec} independently:
    \begin{itemize}
      \item $\sigma' \models n > 0$ is trivial
      \item if $\sigma \models (\sum_{x \in \vec{x}} c_x \cdot (x'-x) + c) \bowtie 0$, then $\sigma' \models (\sum_{x \in \vec{x}} c_x \cdot (x'-x) + n \cdot c) \bowtie 0$
      \item if $\sigma \models e|(\sum_{x \in \vec{x}} c_x \cdot (x'-x) + c)$, then $\sigma' \models e|(\sum_{x \in \vec{x}} c_x \cdot (x'-x) + n \cdot c)$
    \end{itemize}
    Moreover, we have $\sigma \models \mbp(\tau,\sigma,\vec{x})$ and $\sigma \models \mbp(\tau,\sigma,\vec{x}')$ by definition of $\mbp$, and hence we also have $\sigma' \models \mbp(\tau,\sigma,\vec{x})$ and $\sigma' \models \mbp(\tau,\sigma,\vec{x}')$.
    Therefore, $\sigma'$ is a model of $\tip(\tau,\sigma)$, so $\sigma$ is consistent with $\tip(\tau,\sigma)$.

    For item (b), note that $\tip$ only uses the provided model for computing conjunctive variable projections.
    As $\mbp$ has a finite image, the claim follows.

    For item (c), note that the conjuncts $\mbp(\tau,\sigma,\vec{x})$ and $\mbp(\tau,\sigma,\vec{x}')$ of $\tip(\tau,\sigma)$ are irrelevant for transitivity, as they do not relate $\vec{x}$ and $\vec{x}'$.
    Let $\tau' \Def \tip(\tau,\sigma)$.
    It suffices to prove
    \[
      {\to^2_{\tau'}} \subseteq {\to_{\tau'}}.
    \]
    Then the claim follows from a straightforward induction.
    Assume $\theta \models \tau'$, $\theta' \models \tau'$, and
    \[
      \theta(\vec{x}) \to_{\tau'} \theta(\vec{x}') = \theta'(\vec{x}) \to \theta'(\vec{x}').
    \]
    We prove $\hat{\theta} \models \tau'$ where
    \[
      \hat{\theta} \Def [\vec{x} / \theta(\vec{x})] \uplus [\vec{x}' / \theta'(\vec{x}')] \uplus [n / \theta(n) + \theta'(n)].
    \]
    Then the claim follows.
    Again, we consider all literals independently:
    \begin{itemize}
      \item If $\theta \models n > 0$ and $\theta' \models n > 0$, then $\hat{\theta} \models n > 0$.
      \item Consider a literal of the form $\iota \Def t \bowtie 0$ where $t = \sum_{x \in \vec{x}} c_x \cdot (x'-x) + n \cdot c$.
            We have
            \begin{align*}
              \theta(t) = {} & \sum_{x \in \vec{x}} c_x \cdot (\theta(x')-\theta(x)) + \theta(n) \cdot c                                                                               \\
              {} = {}        & \sum_{x \in \vec{x}} c_x \cdot \theta(x') - \sum_{x \in \vec{x}} c_x \cdot \theta(x) + \theta(n) \cdot c                                                \\
              {} = {}        & \sum_{x \in \vec{x}} c_x \cdot \theta'(x) - \sum_{x \in \vec{x}} c_x \cdot \theta(x) + \theta(n) \cdot c \tag{as $\theta(\vec{x}') = \theta'(\vec{x})$}
            \end{align*}
            and
            \begin{align*}
              \theta'(t) = {} & \sum_{x \in \vec{x}} c_x \cdot (\theta'(x')-\theta'(x)) + \theta'(n) \cdot c                                 \\
              {} = {}         & \sum_{x \in \vec{x}} c_x \cdot \theta'(x') - \sum_{x \in \vec{x}} c_x \cdot \theta'(x) + \theta'(n) \cdot c.
            \end{align*}
            Thus, we have:
            \begin{align*}
              \theta(t) + \theta'(t) = {} & -\sum_{x \in \vec{x}} c_x \cdot \theta(x) + \theta(n) \cdot c + \sum_{x \in \vec{x}} c_x \cdot \theta'(x')+ \theta'(n) \cdot c \\
              {} = {}                     & \sum_{x \in \vec{x}} c_x \cdot \theta'(x') - \sum_{x \in \vec{x}} c_x \cdot \theta(x) + (\theta(n) + \theta'(n)) \cdot c       \\
              {} = {}                     & \sum_{x \in \vec{x}} c_x \cdot \hat{\theta}(x') - \sum_{x \in \vec{x}} c_x \cdot \hat{\theta}(x) + \hat{\theta}(n) \cdot c     \\
              {} = {}                     & \hat{\theta}(t)
            \end{align*}
            Therefore, $\theta \models \iota$ and $\theta' \models \iota$ implies $\hat{\theta} \models \iota$ for all ${\bowtie} \in \{\leq,\geq,<,>,=\}$.
      \item Consider a literal of the form $\iota \Def e|t$ where $t = \sum_{x \in \vec{x}} c_x \cdot (x'-x) + n \cdot c$.
            Then we again obtain
            \[
              \theta(t) + \theta'(t) = \hat{\theta}(t)
            \]
            as above.
            Thus, $\theta \models \iota$ and $\theta' \models \iota$ imply $\hat{\theta} \models \iota$.
            \qed
    \end{itemize}
  \end{proof}
}

%% file: unsafety.tex
\section{Proving Unsafety}\label{sec:Unsafety}

We now explain how to adapt \Cref{alg} for also proving unsafety.
Assume that the satisfiability check in \Cref{alg:err2} is successful.
Then an error state is reachable from an initial state via the current trace.
However, the trace may contain learned transitions, so this does not imply unsafety of $\TT$.
The idea for proving unsafety is to replace learned transitions with \emph{accelerated transitions} that result from applying \emph{acceleration techniques}.
\begin{definition}[Acceleration]
  A function $\accel: \QF(\Sigma) \to \QF(\Sigma)$ is called an \emph{acceleration technique} if ${\to_{\accel(\pi)}} \subseteq {\to^+_{\pi}}$ for all relational formulas $\pi$.
\end{definition}
So in contrast to TRL's learned transitions, accelerated transitions under-ap\-prox\-i\-mate transitive closures, and hence they are suitable for proving unsafety.
For arithmetical theories, acceleration techniques are well studied \cite{kroening13,bozga10,acceleration-calculus} (our implementation uses the technique from \cite{acceleration-calculus}).

For any vector of transition formulas $\vec{\rho} = [\rho_1,\ldots,\rho_k]$, let $\compose(\vec{\rho})$ be a transition formula such that ${\to_{\compose(\vec{\rho})}}$ is the composition of the relations ${\to_{\rho_1}}, \ldots, {\to_{\rho_k}}$.
Moreover, for a learned transition $\pi$, we say that $\vec{\tau}_\Loop$ \emph{induced} $\pi$ if we had $\vec{\tau}_\Loop = [\tau_s,\ldots,\tau_{s+\ell-1}]$ in \Cref{alg:loop} when $\pi$ was learned.
Let $\succ$ be the total ordering on the elements of the vector $\vec{\pi}$ from \Cref{alg} with $\pi_i \succ \pi_j$ iff $i > j$, i.e., the input formula $\tau$ is minimal w.r.t.\ $\succ$, and a learned relation is smaller than all relations that were learned later.
Then for vectors of transitions from $\{\mbip(\pi,\sigma) \mid \pi \in \vec{\pi}, \sigma \models \pi\}$, we define the function $\underapprox$ (which yields an \underline{u}nder-\underline{a}pproximation) as follows:
\begin{align*}
  &\underapprox([\eta_1,\ldots,\eta_k]) \Def \compose([\underapprox(\eta_1),\ldots,\underapprox(\eta_k)])\\
  &\underapprox(\eta) \Def
  \begin{cases}
    \eta & \text{if } \eta \models \tau \\
    \accel(\underapprox(\vec{\tau}_\Loop)) & \text{if } \eta \not\models \tau \text{ and}\\
& \phantom{\text{if }} \text{the loop } \vec{\tau}_\Loop \text{ induced }  \underset{\succ}{\min}\{\pi \in \vec{\pi} \mid \eta \models \pi\}
  \end{cases}
\end{align*}
Note that $\underapprox$ is well defined,
 as the following holds for all $\eta' \in \vec{\tau}_\Loop$ in the case $\eta \not\models
 \tau$:
\[
  \min_\succ\{\pi \in \vec{\pi} \mid \eta \models \pi\} \succ \min_\succ\{\pi \in \vec{\pi} \mid \eta' \models \pi\}
\]
The reason is that
 $\min_\succ\{\pi \in \vec{\pi} \mid \eta \models \pi\}$
is induced by $\vec{\tau}_\Loop$, and thus the elements of $\vec{\tau}_\Loop$ are
conjunctive variable projections of $\tau$, or of relations that were learned before
$\min_\succ\{\pi \in \vec{\pi} \mid \eta \models \pi\}$.
Hence, when defining $\underapprox(\eta)$, the ``recursive call''
$\underapprox(\vec{\tau}_\Loop)$ only refers to formulas with smaller index in $\vec{\pi}$, i.e., the
recursion ``terminates''.

So $\underapprox$ leaves original transitions unchanged.
For learned transitions
$\eta$, it first
computes an under-approximation
of the loop $\vec{\tau}_\Loop$ that induced $\eta$, and then it applies acceleration,
resulting in an under-approximation of the transitive closure.
The reason for applying $\underapprox$ before acceleration is that $\vec{\tau}_\Loop$ may again contain learned transitions, which have to be under-approximated first.

To improve \Cref{alg},
instead of returning $\unknown$ in \Cref{alg:err2}, now we first obtain a model $\sigma$ from the SMT solver.
Then we compute the current trace $\vec{\tau} = \trace_{b-1}(\sigma,\vec{\pi}) = [\tau_1,\ldots,\tau_{b-1}]$ (note that the trace has length $b-1$, as $b$ was incremented in \Cref{alg:err1}).
Next, computing $\pi_\underapprox \Def \underapprox(\vec{\tau})$ yields an under-approximation of the states that are reachable with the current trace, but more importantly, $\pi_\underapprox$ also under-approximates $\to^+_\tau$.
The reason is that $\pi_\underapprox$ is constructed from original transitions
by applying $\compose$ (which is exact) and $\accel$ (which yields under-approximations).
Hence, we return $\unsafe$ if there is an initial state $\vec{v}$ and an error state $\vec{v}'$ with $\vec{v} \to_{\pi_\underapprox} \vec{v}'$, i.e., if $\psi_\init \land \pi_\underapprox \land \psi_\err[\vec{x}/\vec{x}']$ is satisfiable.

\begin{example}[Proving Unsafety]
  Consider the relation
  \[
    \underbrace{y > 0 \land x' \doteq x + 1 \land y' \doteq y - 1 \land z' \doteq z}_{\tau_>} {} \lor {} \underbrace{y \doteq 0 \land x' \doteq x \land y' \doteq z \land z' \doteq z}_{\tau_=} \tag{$\tau$}
  \]
  with the initial states given by $\psi_\init \Def x \leq 0$ and the error states given
  by
  $\psi_\err \Def x \geq 1000$.
  Assume that TRL obtains the trace $[\tau_>]$ and learns the relation
  \begin{equation}
    \label{eq:unsafe-learned1}
    n > 0 \land y > 0 \land x' \doteq x + n \land y' \doteq y - n \land z' \doteq z \land y' \geq 0. \tag{$\tau_>^+$}
  \end{equation}
  Next, assume that TRL obtains the trace $[\tau_=,\tau_>^+]$ and learns the relation
  \begin{equation}
    \label{eq:unsafe-learned1}
    y \doteq 0 \land z > 0 \land x' > x \land z > y' \land y' \geq 0 \land z' \doteq z. \tag{$\hat{\tau}^+$}
  \end{equation}
  Note that the transitive closure of the loop $[\tau_=,\tau_>^+]$ cannot be expressed precisely with linear arithmetic, and hence we only obtain the inequations above for $x'$ and $y'$.
  Then an error state is reachable with the trace $[\hat{\tau}^+]$ and we\report{
    \pagebreak[3]} get:
  \begin{align*}
    & \underapprox([\hat{\tau}^+]) \\
    {} = {} & \accel(\underapprox([\tau_=,\tau_>^+])) \tag{as $[\tau_=,\tau_>^+]$ induced $\hat{\tau}^+$} \\
    {} = {} & \accel(\compose(\underapprox(\tau_=),\underapprox(\tau_>^+))) \\
    {} = {} & \accel(\compose(\tau_=,\accel(\underapprox([\tau_>])))) \tag{as $\tau_= \models \tau$ and $[\tau_>]$ induced $\tau_>^+$} \\
    {} = {} & \accel(\compose(\tau_=,\accel(\tau_>))) \tag{as $\tau_> \models \tau$}\\
    {} = {} & \accel(\compose(\tau_=,\tau^+_>)) \tag{as ${\to_{\tau_>}^+} = {\to_{\tau^+_>}}$} \\
    {} = {} & \accel(y \doteq 0 \land n > 0 \land z > 0 \land x' \doteq x + n \land y' \doteq z - n \land z' \doteq z \land y' \geq 0) \\
    {} = {} & y \doteq 0 \land m > 0 \land z > 0 \land x' \doteq x + m \cdot z \land y' \doteq 0 \land z' \doteq z \tag{$\check{\tau}^+$}
  \end{align*}
  For the last step, note that $\check{\tau}^+$ precisely characterizes the transitive closure of $\compose(\tau_=,\tau^+_>)$ for the case $n \doteq z$, and hence it is a valid under-approximation.
  Then a model like $[x/y/0, z/1, m/1000, \; 
 x'/1000,  y'/0,z'/1]$ satisfies $\psi_\init \land \check{\tau}^+ \land \psi_\err[\vec{x}/\vec{x}'] = x \leq 0 \land \check{\tau}^+ \land x' \geq 1000$, which proves unsafety.
\end{example}

%% file: related.tex
\section{Related Work}
\label{sec:related}

Most state-of-the-art infinite state model checking algorithms rely on \emph{interpolation} \cite{spacer,gspacer,tpa-multiloop,imc,lawi,eldarica,pdkind,chatterjee24}.
As discussed in \Cref{sec:intro}, these techniques are very powerful, but also fragile, as interpolation depends on the implementation details of the underlying SMT solver.
In contrast, robustness of TRL depends on the underlying transitive projection.
Our implementation for linear integer arithmetic is a variation of the recurrence analysis from \cite{kincaid15}, which is very robust in our experience.
In particular, our recurrence analysis does not require an SMT solver.

Some state-of-the-art solvers like \textsf{Golem} \cite{golem} also implement \emph{$k$-induction} \cite{kind}, which is more robust than interpolation-based techniques.
However, $k$-induction can only prove $k$-inductive properties, and our experiments (see \Cref{sec:experiments}) show that it is not competitive with other state-of-the-art approaches.

Currently, the most powerful model checking algorithm is \emph{Spacer with global guidance} (GSpacer).
In GSpacer, interpolation is optional, and thus it is more robust than Spacer (without global guidance) and other interpolation-based techniques.
GSpacer is part of the IC3/PDR family of model checking algorithms \cite{ic3}, which differ fundamentally from BMC-based approaches like TRL:
While the latter unroll the transition relation step by step, IC3/PDR proves global properties by combining local reasoning about a single step with induction, interpolation, or other mechanisms like global guidance.
This fundamental difference can also be seen in our evaluation (\Cref{sec:experiments}), where both GSpacer and TRL find many unique safety proofs.
A main advantage of IC3/PDR over BMC-based techniques is scalability, as unrolling the transition relation can quickly become a bottleneck for large systems.
However, TRL mitigates this drawback, as it always backtracks when the trace contains a loop.
Hence, TRL rarely needs to consider long unrollings.

For program verification, there are several techniques that use \emph{transition invariants} or closely related concepts like \emph{loop summarization} or \emph{acceleration} \cite{kincaid15,kroening15,abmc,adcl,silverman19,kincaid24,kroening13,bozga10,FlatFramework}.

Acceleration-based techniques \cite{kroening15,abmc,adcl,bozga10,FlatFramework} try to compute transitive closures of loops exactly, or to under-ap\-prox\-i\-mate them.
However, even for simple loops like {\tt while(...) x = x + y}, meaningful under-ap\-prox\-i\-ma\-tions cannot be expressed in a decidable logic.
For example, for the loop above, the non-linear formula $\exists n.\ n > 0 \land x' \doteq x + n \cdot y$ would be a precise characterization of the transitive closure.
Thus, these approaches either only use acceleration for loops where decidable logics are sufficient, or they resort to undecidable logics, which is challenging for the underlying SMT solver.
TRL is not restricted to under-ap\-prox\-i\-ma\-tions, and hence it can ap\-prox\-i\-mate arbitrary loops without resorting to undecidable logics.

As mentioned in \Cref{sec:intro}, TRL has been inspired by our work on accelerated bounded
model checking (ABMC) \cite{abmc}.
For the reasons explained above, TRL is more powerful than ABMC for proving safety (see also \Cref{sec:experiments}), but it does not subsume ABMC for that purpose.
The reason is that ABMC's acceleration sometimes yields non-linear relations, whereas TRL only learns linear relations.
On the other hand, ABMC is more powerful than TRL for proving unsafety, as TRL only tries to construct a counterexample from a single trace, whereas ABMC explores many different traces.
Hence, both approaches are orthogonal in practice.
Apart from that, a major difference is that ABMC uses blocking clauses only sparingly, which is disadvantageous for proving safety.
The reason is that blocking clauses reduce the search space for counterexamples, and to prove safety, this search space has to be exhausted, eventually.
However, using them whenever possible is disadvantageous as well, as it blows up the underlying SMT problem unnecessarily.
Hence, TRL uses them on demand instead.
Other important conceptual differences include TRL's use of: model based projection, which
cannot be used by ABMC, as acceleration yields formulas in theories without effective
quantifier elimination; backtracking, which is useful for proving safety as it keeps the
underlying SMT problem small, but it would slow down the discovery of
counterexamples in
ABMC.  

Our tool \tool{LoAT} also implements ADCL \cite{adcl}, which embeds acceleration into a depth-first search for counterexamples (whereas ABMC and TRL perform breadth-first search).
Thus, TRL is conceptually closer to ABMC, and for proving unsafety, ABMC is superior to ADCL (see \cite{abmc}).
However, as witnessed by the results of the annual \emph{Termination Competition} \cite{termcomp}, ADCL is currently the most powerful technique for proving non-termination of transition systems (see \cite{adcl-nt}).

In contrast to acceleration-based techniques, recurrence analysis and loop summarization \cite{kincaid15,silverman19,kincaid24,kroening13} over-ap\-prox\-i\-mate transitive closures of loops.
TRL is not restricted to over-ap\-prox\-i\-ma\-tions, which allows for increased precision, as our transitive projections can perform a case analysis based on the provided model.
Note that this cannot be done in a framework that requires over-ap\-prox\-i\-ma\-tions, as $\tip$ cannot be used to compute over-approximations, in general (see \Cref{remark:properties-tip}).

The techniques from \cite{silverman19,kincaid24} express over-ap\-prox\-i\-ma\-tions using (variations of) \emph{vector addition systems with resets (VASRs)}, which are less expressive than our transitive projections, as the latter may yield arbitrary relational formulas.
However, \cite{silverman19,kincaid24} ensure that the computed ap\-prox\-i\-ma\-tions are the best VASR-ap\-prox\-i\-ma\-tions.
Our approach does not provide any such guarantees, so both approaches are orthogonal.

Apart from these techniques, transition invariants have mainly been used in the context of termination analysis and temporal verification \cite{tsitovich11,compositional_transition_invariants,heizmann10,podelski11,zuleger18}.
We expect that our approach can be adapted to these problems, but it is unclear whether it would be competitive.

%% file: experiments.tex
\section{Experiments and Conclusion}
\label{sec:experiments}

We implemented TRL in our tool \tool{LoAT} \cite{loat}.
Our implementation uses the SMT solvers \tool{Yices} \cite{yices} and \tool{Z3} \cite{z3}, and currently, it is restricted to linear integer arithmetic.
We evaluated our approach on all examples from the category LIA-Lin (linear CHCs with linear integer arithmetic) from the CHC competitions~2023 and 2024 \cite{CHC-COMP} (excluding duplicates), resulting in a collection of 626 problems from numerous applications like verification of \href{https://github.com/chc-comp/hcai-bench}{\tool{C}}, \href{https://github.com/chc-comp/rust-horn}{\tool{Rust}}, \href{https://github.com/chc-comp/jayhorn-benchmarks}{\tool{Java}}, and \href{https://github.com/chc-comp/hopv}{higher-order} programs, and \href{https://github.com/mattulbrich/llreve}{regression verification of \tool{LLVM} programs},
see \cite{chc-comp23,chc-comp24} for details.
By using CHCs as input format, our approach can be used by any CHC-based tool like \tool{Korn} \cite{korn} and \tool{SeaHorn} \cite{seahorn} for \pl{C} and \CXX{} programs, \tool{JayHorn} for \pl{Java} programs \cite{jayhorn}, \tool{HornDroid} for \pl{Android} \cite{horndroid}, \tool{RustHorn} for \pl{Rust} programs \cite{rusthorn}, and \tool{SmartACE} \cite{smartACE} and \tool{SolCMC} \cite{solcmc} for \pl{Solidity}.

We compared TRL with the techniques of leading CHC solvers.
More precisely, we evaluated the following configurations:
\begin{description}
  \item[\tool{LoAT}] We used \tool{LoAT}'s implementations of TRL (\tool{LoAT TRL}), ABMC \cite{abmc}
    (\tool{LoAT ABMC}), and \emph{$k$-induction} \cite{kind} (\tool{LoAT KIND}).
  \item[\tool{Z3} \cite{z3}] We used \tool{Z3} 4.13.3, where we evaluated its implementations of the Spacer \cite{spacer} (\tool{Z3 Spacer}) and GSpacer \cite{gspacer} (\tool{Z3 GSpacer}) algorithms, and of BMC (\tool{Z3 BMC}).
  \item[\tool{Golem} \cite{golem}] We used \tool{Golem} 0.6.2, where we evaluated its implementations of \emph{transition power abstraction}~\cite{tpa-multiloop} (\tool{Golem TPA}), \emph{interpolation based model checking} \cite{imc} (\tool{Golem IMC}), \emph{lazy abstraction with interpolants} \cite{lawi} (\tool{Golem LAWI}), \emph{predicate abstraction} \cite{predicate_abstraction} with \emph{CEGAR} \cite{cegar} (\tool{Golem PA}), \emph{property directed $k$-induction} \cite{pdkind} (\tool{Golem PDKIND}), Spacer (\tool{Golem Spacer}), and BMC (\tool{Golem BMC}).
        Note that \tool{Golem}'s implementation of $k$-induction only applies to linear CHC problems with a specific structure (just one fact, one recursive rule, and one query), so we excluded it from our evaluation.
  \item[\tool{Eldarica} \cite{eldarica}] We used the default configuration of \tool{Eldarica} 2.1, which is based on predicate abstraction and CEGAR.
\end{description}
We ran our experiments on \href{https://help.itc.rwth-aachen.de/service/rhr4fjjutttf/article/fbd107191cf14c4b8307f44f545cf68a/}{CLAIX-2023-HPC nodes} of the \paper{\href{https://help.itc.rwth-aachen.de/service/rhr4fjjutttf/}{RWTH Uni\-ver\-si\-ty High Performance Computing Cluster}}\report{RWTH Uni\-ver\-si\-ty High Performance Computing Cluster\footnote{\url{https://help.itc.rwth-aachen.de/service/rhr4fjjutttf}}} with a memory limit of 10560 MiB ($\approx$~11GB) and a timeout of 300~s per example.

\begin{figure}[th!]
  \begin{minipage}{0.53\textwidth}
    \begin{tabular}{|c||c||c|c||c|c|}
      \hline                       & \multirow{2}{*}{\checkmark} & \multicolumn{2}{c||}{$\safe$} & \multicolumn{2}{c|}{$\unsafe$} \\
      \hhline{~~----}              &                 & \checkmark                    & {\bf !}   & \checkmark      & {\bf !} \\
      \hline\hline \tool{LoAT TRL} & \underline{386} & 267                           & 16(19)    & 119             & 1(9)    \\
      \hline \tool{Z3 GSpacer}     & 373             & \underline{287}               & 16(16)    & 86              & 3(3)    \\
      \hline \tool{Golem Spacer}   & 340             & 231                           & --        & 109             & --      \\
      \hline \tool{Z3 Spacer}      & 339             & 236                           & --        & 103             & --      \\
      \hline \tool{Golem IMC}      & 318             & 215                           & 2(2)      & 103             & 0(0)    \\
      \hline \tool{Golem PDKIND}   & 315             & 224                           & 2(2)      & 91              & 0(0)    \\
      \hline \tool{Eldarica}       & 298             & 219                           & 7(7)      & 79              & 0(0)    \\
      \hline \tool{Golem LAWI}     & 290             & 182                           & 0(0)      & 108             & 1(1)    \\
      \hline \tool{Golem TPA}      & 289             & 170                           & 1(2)      & 119             & 2(4)    \\
      \hline \tool{LoAT ABMC}      & 272             & 142                           & 1(--)     & \underline{130} & 2(--)   \\
      \hline \tool{Golem PA}       & 253             & 160                           & --        & 93              & --      \\
      \hline \tool{LoAT KIND}      & 245             & 133                           & 0(0)      & 112             & 0(0)    \\
      \hline \tool{Golem BMC}      & 145             & 34                            & 0(0)      & 111             & 0(0)    \\
      \hline \tool{Z3 BMC}         & 145             & 32                            & --        & 113             & --      \\
      \hline
    \end{tabular}
    \label{tab}
  \end{minipage}
  \hspace{-1.5em}
  \begin{minipage}{0.49\textwidth}
    \begin{tikzpicture}[scale=0.73]
      \begin{axis}[
        legend pos=south east,
        ylabel=solved instances,
        y label style={at={(axis description cs:-0.07,.5)},anchor=south},
        y tick label style={rotate=90},
        xticklabel={$\pgfmathprintnumber{\tick}$s},
        ymin=230,
        ymax=390,
        xmin=-10,
        xmax=300,
        legend columns=2]
        \addplot[color=black,solid,thick] table[col sep=comma,header=false,x index=0,y index=1] {til.csv}; \addlegendentry{\tool{\scriptsize TRL}}
        \addplot[color=violet,densely dashed,thick] table[col sep=comma,header=false,x index=0,y index=1] {gspacer.csv}; \addlegendentry{\tool{\scriptsize GSpacer}}
        \addplot[color=red,dashed,thick] table[col sep=comma,header=false,x index=0,y index=1] {spacer.csv}; \addlegendentry{\tool{\scriptsize Golem Spacer}}
        \addplot[color=blue,densely dotted,thick] table[col sep=comma,header=false,x index=0,y index=1] {z3spacer.csv}; \addlegendentry{\tool{\scriptsize Z3 Spacer}}
        \addplot[color=green,loosely dashed,thick] table[col sep=comma,header=false,x index=0,y index=1] {imc.csv}; \addlegendentry{\tool{\scriptsize IMC}}
        \addplot[color=lightgray,dotted,thick] table[col sep=comma,header=false,x index=0,y index=1] {pdkind.csv}; \addlegendentry{\tool{\scriptsize PDKIND}}
        \addplot[color=purple,loosely dotted,thick] table[col sep=comma,header=false,x index=0,y index=1] {eld.csv}; \addlegendentry{\tool{\scriptsize Eldarica}}
      \end{axis}
    \end{tikzpicture}
  \end{minipage}
\end{figure}

\noindent
The results can be seen in the table \vpageref[above]{tab}.
The first column marked with \checkmark contains the number of solved
instances (i.e., all remaining instances could not be solved by the respective configuration).
The columns with {\bf !} show the num\-ber of unique proofs, i.e., the number of examples that could only be solved by the corresponding configuration.
Such a comparison only makes sense if just one implementation of each algorithm is considered.
For instance, GSpacer is an improved version of Spacer, so \tool{Z3 GSpacer}, \tool{Z3 Spacer}, and \tool{Golem Spacer} work well on the same class of examples.
Hence, if all of them were considered, they would find very few unique proofs.
The same is true for \tool{Eldarica} and \tool{Golem PA} as well as \tool{Golem BMC} and \tool{Z3 BMC}.
Thus, for {\bf !} we disregarded \tool{Z3 Spacer}, \tool{Golem Spacer}, \tool{Golem PA}, and \tool{Z3 BMC}.

The numbers in parentheses in the columns marked with {\bf !} show the number of unique proofs when also disregarding \tool{LoAT ABMC}.
We included this number since TRL was inspired by ABMC.
So analogously to GSpacer and Spacer, one could consider TRL to be an improved version of ABMC for proving safety.

The table shows that TRL is highly competitive:
Overall, it solves the most instances.
Even more importantly, it finds many unique proofs, i.e., it is orthogonal to existing model checking algorithms, and hence it improves the state of the art.

Regarding safe instances, TRL is only outperformed by GSpacer.
Regarding unsafe instances, TRL is only outperformed by ABMC, which has specifically been
developed for finding (deep) counterexamples, whereas TRL's main purpose is proving
safety.

The plot on the right shows how many instances were solved within 300~s, where we only include the seven best configurations for readability.
It shows that TRL is highly competitive, not only in terms of solved examples, but also in terms of runtime.

Our implementation is open-source and available on Github \cite{loat-web}.
For the sources, a pre-compiled binary, and more information on our evaluation, we refer to \cite{website}.

\subsubsection*{Conclusion}

We presented \emph{Transitive Relation Learning (TRL)}, a powerful model checking algorithm for infinite state systems.
Instead of searching for invariants, TRL adds transitive relations to the analyzed system until its \emph{diameter} (the number of steps that is required to cover all reachable states) becomes finite, which facilitates a safety proof.
As it does not search for invariants, TRL does not need interpolation, which is in contrast to most other state-of-the-art techniques.
Nevertheless, our evaluation shows that TRL is highly competitive with interpolation-based approaches.
Moreover, not using interpolation allows us to avoid the well-known fragility of interpolation-based approaches.
Finally, integrating \emph{acceleration techniques} into TRL also yields a competitive technique for proving \emph{un}safety.

In future work, we plan to support other theories like reals, bitvectors, and arrays, and
we will investigate an extension to temporal verification.

%% file: proofs.tex
\begin{appendix}
  \appendixproofsection{Appendix}\label{sec:proofs}
  \appendixproof*{thm:soundness}
  \appendixproof*{thm:tip}
\end{appendix}

%% file: main.bbl
\providecommand{\noopsort}[1]{}
\begin{thebibliography}{10}
\providecommand{\url}[1]{\texttt{#1}}
\providecommand{\urlprefix}{URL }
\providecommand{\doi}[1]{https://doi.org/#1}

\bibitem{solcmc}
Alt, L., Blicha, M., Hyv{\"{a}}rinen, A.E.J., Sharygina, N.: \tool{SolCMC}:
  \pl{Solidity} compiler's model checker. In: CAV~'22. pp. 325--338. LNCS 13371
  (2022). \doi{10.1007/978-3-031-13185-1\_16}

\bibitem{FlatFramework}
Bardin, S., Finkel, A., Leroux, J., Schnoebelen, P.: Flat acceleration in
  symbolic model checking. In: ATVA~'05. pp. 474--488. LNCS 3707 (2005).
  \doi{10.1007/11562948_35}

\bibitem{smtlib}
Barrett, C., Fontaine, P., Tinelli, C.: {The Satisfiability Modulo Theories
  Library (SMT-LIB)}. {\tt www.SMT-LIB.org} (2016)

\bibitem{bmc2}
Biere, A., Cimatti, A., Clarke, E.M., Strichman, O., Zhu, Y.: Bounded model
  checking. Advances in Computers  \textbf{58},  117--148 (2003).
  \doi{10.1016/S0065-2458(03)58003-2}

\bibitem{golem}
Blicha, M., Britikov, K., Sharygina, N.: The {{\tool{{Golem}}}} {Horn} solver.
  In: CAV~'23. pp. 209--223. LNCS 13965 (2023).
  \doi{10.1007/978-3-031-37703-7\_10}

\bibitem{bozga10}
Bozga, M., Iosif, R., Kone\v{c}n\'{y}, F.: Fast acceleration of ultimately
  periodic relations. In: CAV~'10. pp. 227--242. LNCS 6174 (2010).
  \doi{10.1007/978-3-642-14295-6\_23}

\bibitem{ic3}
Bradley, A.R.: {SAT}-based model checking without unrolling. In: {VMCAI} '11.
  pp. 70--87. LNCS 6538 (2011). \doi{10.1007/978-3-642-18275-4\_7}

\bibitem{dd-exp}
Bremner, D.: Incremental convex hull algorithms are not output sensitive.
  Discret. Comput. Geom.  \textbf{21}(1),  57--68 (1999).
  \doi{10.1007/PL00009410}

\bibitem{tpa-multiloop}
Britikov, K., Blicha, M., Sharygina, N., Fedyukovich, G.: Reachability analysis
  for multiloop programs using transition power abstraction. In: FM~'24. pp.
  558--576. LNCS 14933 (2024). \doi{10.1007/978-3-031-71162-6\_29}

\bibitem{horndroid}
Calzavara, S., Grishchenko, I., Maffei, M.: \tool{HornDroid}: Practical and
  sound static analysis of \tool{Android} applications by {SMT} solving. In:
  EuroS{\&}P~'16. pp. 47--62. {IEEE} (2016). \doi{10.1109/EuroSP.2016.16}

\bibitem{CHC-COMP}
{CHC Competition}, \url{https://chc-comp.github.io}

\bibitem{cegar}
Clarke, E.M., Grumberg, O., Jha, S., Lu, Y., Veith, H.: Counterexample-guided
  abstraction refinement. In: CAV~'00. pp. 154--169. LNCS 1855 (2000).
  \doi{10.1007/10722167\_15}

\bibitem{cooper72}
Cooper, D.C.: Theorem proving in arithmetic without multiplication. Machine
  Intelligence  \textbf{7},  91--99 (1972)

\bibitem{chc-comp23}
{De Angelis}, E., {Govind V K}, H.: {CHC-COMP} 2023: Competition report (2023),
  \url{https://chc-comp.github.io/2023/CHC_COMP_2023_Competition_Report.pdf}

\bibitem{yices}
Dutertre, B.: \tool{Yices} 2.2. In: CAV~'14. pp. 737--744. LNCS 8559 (2014).
  \doi{10.1007/978-3-319-08867-9\_49}

\bibitem{enderton}
Enderton, H.B.: A Mathematical Introduction to Logic. Academic Press (1972)

\bibitem{korn}
Ernst, G.: Loop verification with invariants and contracts. In: VMCAI~'22. pp.
  69--92. LNCS 13182 (2022). \doi{10.1007/978-3-030-94583-1\_4}

\bibitem{chc-comp24}
Ernst, G., Morales, J.F.: {CHC-COMP} 2024: Competition report (2024),
  \url{https://chc-comp.github.io/2024/CHC-COMP%202024%20Report%20-%20HCSV.pdf}

\bibitem{kincaid15}
Farzan, A., Kincaid, Z.: Compositional recurrence analysis. In: FMCAD~'15. pp.
  57--64 (2015). \doi{10.1109/FMCAD.2015.7542253}

\bibitem{acceleration-calculus}
Frohn, F.: A calculus for modular loop acceleration. In: TACAS~'20. pp. 58--76.
  LNCS 12078 (2020). \doi{10.1007/978-3-030-45190-5\_4}

\bibitem{loat}
Frohn\noopsort{4}, F., Giesl, J.: Proving non-termination and lower runtime
  bounds with \tool{LoAT} (system description). In: IJCAR~'22. pp. 712--722.
  LNCS 13385 (2022). \doi{10.1007/978-3-031-10769-6\_41}

\bibitem{adcl-nt}
Frohn\noopsort{5}, F., Giesl, J.: Proving non-termination by {Acceleration
  Driven Clause Learning}. In: CADE~'23. pp. 220--233. LNCS 14132 (2023).
  \doi{10.1007/978-3-031-38499-8\_13}

\bibitem{adcl}
Frohn\noopsort{6}, F., Giesl, J.: {ADCL}: {A}cceleration {D}riven {C}lause
  {L}earning for constrained {H}orn clauses. In: SAS~'23. pp. 259--285. LNCS
  14284 (2023). \doi{10.1007/978-3-031-44245-2\_13}

\bibitem{abmc}
Frohn\noopsort{6}, F., Giesl, J.: Integrating loop acceleration into bounded
  model checking. In: FM~'24. pp. 73--91. LNCS 14933 (2024).
  \doi{10.1007/978-3-031-71162-6\_4}

\bibitem{arxiv}
Frohn\noopsort{7}, F., Giesl, J.: Infinite state model checking by learning
  transitive relations. CoRR  \textbf{abs/2502.04761} (2025).
  \doi{10.48550/arXiv.2502.04761}

\bibitem{website}
Frohn\noopsort{8}, F., Giesl, J.: Evaluation of ``{I}nfinite {S}tate {M}odel
  {C}hecking by {L}earning {T}ransitive {R}elations'' (2025),
  \url{https://loat-developers.github.io/trl-eval/}

\bibitem{loat-web}
Frohn\noopsort{9}, F., Giesl, J.: \tool{LoAT} website,
  \url{https://loat-developers.github.io/LoAT/}

\bibitem{convex-hull}
Genov, B.: The Convex Hull Problem in Practice: Improving the Running Time of
  the Double Description Method. Ph.D. thesis, Bremen University, Germany
  (2015), \url{https://media.suub.uni-bremen.de/handle/elib/833}

\bibitem{termcomp}
Giesl, J., Rubio, A., Sternagel, C., Waldmann, J., Yamada, A.: The termination
  and complexity competition. In: TACAS~'19. pp. 156--166. LNCS 11429 (2019).
  \doi{10.1007/978-3-030-17502-3_10}

\bibitem{predicate_abstraction}
Graf, S., Sa{\"{\i}}di, H.: Construction of abstract state graphs with
  \tool{PVS}. In: CAV~'97. pp. 72--83. LNCS 1254 (1997).
  \doi{10.1007/3-540-63166-6\_10}

\bibitem{seahorn}
Gurfinkel, A., Kahsai, T., Komuravelli, A., Navas, J.A.: The \tool{SeaHorn}
  verification framework. In: CAV~'15. pp. 343--361. LNCS 9206 (2015).
  \doi{10.1007/978-3-319-21690-4\_20}

\bibitem{heizmann10}
Heizmann, M., Jones, N.D., Podelski, A.: Size-change termination and transition
  invariants. In: Cousot, R., Martel, M. (eds.) SAS~'10. pp. 22--50. LNCS 6337
  (2010). \doi{10.1007/978-3-642-15769-1\_4}

\bibitem{eldarica}
Hojjat, H., R{\"{u}}mmer, P.: The \tool{Eldarica} {Horn} solver. In: FMCAD~'18.
  pp.~1--7 (2018). \doi{10.23919/FMCAD.2018.8603013}

\bibitem{pdkind}
Jovanovic, D., Dutertre, B.: Property-directed k-induction. In: {FMCAD}~'16.
  pp. 85--92 (2016). \doi{10.1109/FMCAD.2016.7886665}

\bibitem{jayhorn}
Kahsai, T., R{\"{u}}mmer, P., Sanchez, H., Sch{\"{a}}f, M.: \tool{JayHorn}: {A}
  framework for verifying \pl{Java} programs. In: CAV~'16. pp. 352--358. LNCS
  9779 (2016). \doi{10.1007/978-3-319-41528-4\_19}

\bibitem{kincaid17}
Kincaid, Z., Breck, J., Boroujeni, A.F., Reps, T.W.: Compositional recurrence
  analysis revisited. In: PLDI~'17. pp. 248--262 (2017).
  \doi{10.1145/3062341.3062373}

\bibitem{spacer}
Komuravelli, A., Gurfinkel, A., Chaki, S.: {SMT}-based model checking for
  recursive programs. Formal Methods Syst. Des.  \textbf{48}(3),  175--205
  (2016). \doi{10.1007/s10703-016-0249-4}

\bibitem{octagonsP}
Kone\v{c}n\'{y}, F.: {PTIME} computation of transitive closures of octagonal
  relations. In: TACAS~'16. pp. 645--661. LNCS 9636 (2016).
  \doi{10.1007/978-3-662-49674-9\_42}

\bibitem{gspacer}
Krishnan, H.G.V., Chen, Y., Shoham, S., Gurfinkel, A.: Global guidance for
  local generalization in model checking. In: CAV~'20. pp. 101--125. LNCS 12225
  (2020). \doi{10.1007/978-3-030-53291-8\_7}

\bibitem{kroening13}
Kroening, D., Sharygina, N., Tonetta, S., Tsitovich, A., Wintersteiger, C.M.:
  Loop summarization using state and transition invariants. Formal Methods
  Syst. Des.  \textbf{42}(3),  221--261 (2013). \doi{10.1007/S10703-012-0176-Y}

\bibitem{compositional_transition_invariants}
Kroening, D., Sharygina, N., Tsitovich, A., Wintersteiger, C.M.: Termination
  analysis with compositional transition invariants. In: CAV~'10. pp. 89--103.
  LNCS~6174 (2010). \doi{10.1007/978-3-642-14295-6\_9}

\bibitem{kroening15}
Kroening\noopsort{5}, D., Lewis, M., Weissenbacher, G.: Proving safety with
  trace automata and bounded model checking. In: {FM}~'15. pp. 325--341. LNCS
  9109 (2015). \doi{10.1007/978-3-319-19249-9\_21}

\bibitem{rusthorn}
Matsushita, Y., Tsukada, T., Kobayashi, N.: \tool{RustHorn}: {CHC}-based
  verification for \pl{Rust} programs. {ACM} Trans. Program. Lang. Syst.
  \textbf{43}(4),  15:1--15:54 (2021). \doi{10.1145/3462205}

\bibitem{imc}
McMillan, K.L.: Interpolation and {{SAT}}-based model checking. In: CAV~'03.
  pp. 1--13. LNCS 2725 (2003). \doi{10.1007/978-3-540-45069-6\_1}

\bibitem{lawi}
McMillan, K.L.: Lazy abstraction with interpolants. In: CAV~'06. pp. 123--136.
  LNCS 4144 (2006). \doi{10.1007/11817963\_14}

\bibitem{z3}
\noopsort{Moura}{de Moura}, L., Bj{\o}rner, N.: \tool{Z3}: An efficient {SMT}
  solver. In: TACAS\ '08. pp. 337--340. LNCS 4963 (2008).
  \doi{10.1007/978-3-540-78800-3\_24}

\bibitem{kincaid24}
Pimpalkhare, N., Kincaid, Z.: Semi-linear {VASR} for over-approximate
  semi-linear transition system reachability. In: RP~'24. pp. 154--166. LNCS
  15050 (2024). \doi{10.1007/978-3-031-72621-7\_11}

\bibitem{transition_invariants}
Podelski, A., Rybalchenko, A.: Transition invariants. In: LICS~'04. pp. 32--41
  (2004). \doi{10.1109/LICS.2004.1319598}

\bibitem{podelski11}
Podelski, A., Rybalchenko, A.: Transition invariants and transition predicate
  abstraction for program termination. In: TACAS~'11. pp. 3--10. LCNS 6605
  (2011). \doi{10.1007/978-3-642-19835-9\_2}

\bibitem{fmplex}
Promies, V., {\'{A}}brah{\'{a}}m, E.: A divide-and-conquer approach to variable
  elimination in linear real arithmetic. In: FM~'24. pp. 131--148. LNCS 14933
  (2024). \doi{10.1007/978-3-031-71162-6\_7}

\bibitem{kind}
Sheeran, M., Singh, S., St{\aa}lmarck, G.: Checking safety properties using
  induction and a {SAT}-solver. In: FMCAD~'00. pp. 108--125. LNCS 1954 (2000).
  \doi{10.1007/3-540-40922-X\_8}

\bibitem{silverman19}
Silverman, J., Kincaid, Z.: Loop summarization with rational vector addition
  systems. In: CAV~'19. pp. 97--115. LNCS 11562 (2019).
  \doi{10.1007/978-3-030-25543-5\_7}

\bibitem{chatterjee24}
Solanki, M., Chatterjee, P., Lal, A., Roy, S.: Accelerated bounded model
  checking using interpolation based summaries. In: TACAS~'24. pp. 155--174.
  LNCS 14571 (2024). \doi{10.1007/978-3-031-57249-4\_8}

\bibitem{tsitovich11}
Tsitovich, A., Sharygina, N., Wintersteiger, C.M., Kroening, D.: Loop
  summarization and termination analysis. In: TACAS~'11. pp. 81--95. LNCS 6605
  (2011)

\bibitem{smartACE}
Wesley, S., Christakis, M., Navas, J.A., Trefler, R.J., W{\"{u}}stholz, V.,
  Gurfinkel, A.: Verifying \pl{Solidity} smart contracts via communication
  abstraction in \tool{{\mbox{SmartACE}}}. In: VMCAI~'22. pp. 425--449. LNCS
  13182 (2022). \doi{10.1007/978-3-030-94583-1\_21}

\bibitem{zuleger18}
Zuleger, F.: Inductive termination proofs with transition invariants and their
  relationship to the size-change abstraction. In: SAS~'18. pp. 423--444. LNCS
  11002 (2018). \doi{10.1007/978-3-319-99725-4\_25}

\end{thebibliography}
